\documentclass[12pt]{article}
\usepackage{hyperref} 
\usepackage{array} 
\usepackage{graphicx} 
\usepackage{amsthm}
\usepackage{amsmath}
\usepackage{natbib}
\usepackage{subfigure}

\usepackage{amstext}
\usepackage{amssymb}
\def\dirac{\mathrm{\boldsymbol{\delta}}}
\def\R{\mathbb{R}}
\def\Z{\mathbb{Z}}

\theoremstyle{plain}% Theorem-like structures
\newtheorem{theorem}{Theorem}[section]

\newtheorem{corollary}[theorem]{Corollary}

\newtheorem{lemma}[theorem]{Lemma}
\newtheorem{proposition}[theorem]{Proposition}
\newtheorem{xmpl}[theorem]{Exapmle}

\theoremstyle{definition}

\theoremstyle{remark}
\newtheorem{remark}{Remark}
\newcolumntype{C}[1]{>{\centering\let\newline\\\arraybackslash\hspace{0pt}}m{#1}}

\begin{document}

\title{Estimation of Zero Intelligence Models by L1 Data}

\author{MARTIN \v SM\' ID\thanks{Department of Econometrics, Institute of Information Theory and Automation of the ASCR, Pod Vod\' arenskou v\v e\v z\' \i 4, Praha 8, CZ18208, Czech Republic, tel:+420 26605 2411} \thanks{Corresponding author: email:smid@utia.cas.cz}
\thanks{This work was supported by the Czech Science Foundation under Grant GAP402/12/G097.}
}

\maketitle

\begin{abstract} A unit volume zero intelligence (ZI) model is defined and the distribution of its L1 process is recursively described. Further, a generalized ZI (GZI) model allowing non-unit market orders, shifts of quotes and general in-spread events is proposed and a formula for the conditional distribution of its quotes is given, together with a formula for price impact. For both the models, MLE estimators are formulated and shown to be consistent and asymptotically normal. Consequently, the estimators are applied to data of six US stocks from nine electronic markets. It is found that more complex variants of the models, despite being significant, do not give considerably better predictions than their simple versions with constant intensities.

\end{abstract}

\noindent Keywords: Limit Order Market, Stochastic Models, Econometric Methods, 

\newpage

\section{Introduction}
With the recent wide expansion of trading according to continuous double auction (CDA), the importance of mathematical modelling of this trading mechanism grows. 

A number of models of CDA exist which assume a rational behaviour of agents involved (see, e.g., \cite{Parlour08} and the references therein); these models are, however, dependent on many arbitrary assumptions and do not give much better empirical results than models assuming a purely random behaviour of agents (see \cite{Gode93} for a discussion). Thus, as an alternative to the ``rational'' approach, zero intelligence (ZI) models of the CDA started to be studied: out of a number of similar models of that kind, let us name \cite{Stigler64, Maslov00, Challet01, Luckock03, Smith03, Mike08} or \citet{Cont10}. All these models assume unit order sizes, Poisson arrivals of market and limit orders and locally constant cancellation rates depending on the distance to the quotes.\footnote{The only exceptions are \cite{Maslov00} and \cite{Challet01} with a discrete time (which could, however, be regarded as Poisson events' time) and \cite{Mike08}, in which the cancellation rate also depends on the order books' size and imbalance} 
%Some of the models work with continuous prices %(\cite{Maslov00,Luckock03}), the rest of them %assume discrete (tick) prices. 
For a survey of the ZI models and their characteristics, see \cite{chakraborti2011econophysicsi} and \cite{chakraborti2011econophysicsii}, especially Section II of the latter paper. For the recent developments, see \cite{abergel2011econophysics}.

Although, according to their advocates, the ZI models are able to mimic many stylised empirical facts such as fat tails or non-Gaussianity (\cite{Slanina01,Slanina08,Smid08d, Cont11}), no rigorous statistical evidence in that respect has been presented yet due to intractability of these models. In \cite{Smid12}, a conditional distribution given L1 data is described; however, the model considered in this paper is too general to ensure consistency of statistical estimators.

The present paper simplifies the approach of \cite{Smid12} by assuming only a finite number of prices so that it is possible to construct consistent asymptotically normal estimators of the models' parameters.

The exposition is started by introducing a sufficiently general setting covering a wide range of existing zero intelligence (ZI) models, e.g., an early work of \cite{Stigler64}, the model by \cite{Cont10}, a discretised version of \cite{Luckock03} and a slightly modified version of \cite{Smith03}. After demonstrating ergodicity of the covered models, a conditional density of jumps of the L1 process (i.e., the bid, the ask, and the corresponding offered volumes) given its history is formulated.\footnote{Although order book (L2) data  might seem more natural to be used for the estimation, the better availability and lower price of L1 data speak for using them. This approach, in addition, allows us to avoid a problem of hidden limit orders, which are invisible in the books but affect sizes of quotes' jumps.}

Further, in order to treat the most obvious discrepancies between the ZI models and reality, a generalised (GZI) setting is proposed, which allows shifts of the quotes, non-unit market orders and general distributions of inside-the-spread events. A formula for a conditional density of the out-of-spread jumps of the quotes, later used in estimating the in-book parameters, is given.

Subsequently, Maximum Likelihood estimators are formulated and several variants of both the ZI and GZI models are estimated by means of L1 data for six stocks on nine US electronic markets. Further, for each stock-market pair, the variants of the models are compared by their ability to forecast magnitudes of the quotes' jumps. It is found that more complicated variants of the models do not bring significant improvements in comparison with the simple variant with constant intensities, which itself, however, nearly always performs significantly better than a naive prediction of the jumps. It is also found that GZI variants of the models do not perform significantly better than their ZI counterparts.
%; however, as only the GZI setting is capable of %predicting market impact, it is worthwhile to %use it in practice rather than the ZI one.

This paper is organised as follows: First, the ZI setting is defined (Section \ref{sec:zi}) and the distribution of the L1 process given its history is derived. Consequently, the GZI setting is introduced and a formula for the density of the out-of-spread jumps and volume changes is given (Section \ref{sec:gzi}). Finally, the empirical evidence is discussed (Section \ref{sec:empirical}) and the paper is concluded (Section \ref{sec:conclusion}). The Appendix includes a proof of Theorem \ref{th:dist} (Appendix \ref{sec:proofs}), a proof of asymptotic properties of the MLE estimator (Appendix \ref{sec:mle}), and detailed results on the estimation (Appendix \ref{sec:detail}).

\section{Zero Intelligence Model}
\label{sec:zi}
\global\long\def\P{\mathbb{P}}
\global\long\def\E{\mathbb{E}}
\global\long\def\N{\mathbb{N}}
\global\long\def\indep#1{{\perp\hspace{-2mm}\perp}#1}
\global\long\def\Po{\mathrm{Po}}
\global\long\def\Bi{\mathrm{Bi}}

\label{existtheory}

\subsection{Definition}
\label{ss:zi}
Consider a general discrete-price continuous-time zero intelligence model with unit order sizes described by a pure jump type process
$$
\Xi_t=(A_t,B_t),\qquad t\geq0,
$$
where
 \[
A_{t}=(A^1_{t}, A^2_{t},\dots,A^n_{t}) \in \N^n_0,
\]
and  
\[
B_{t}=(B^1_{t}, B^2_{t},\dots,B^n_{t}) \in \N^n_0,
\]
are the sell limit order book and buy limit order book, respectively; here, $n$ is a number of possible prices (ticks) and, for any $p$,
$A^p_{t}$ ($B^p_{t}$) stands for the number of the sell (buy) limit
orders with price $p$ waiting at time $t$. Further, denote by
\[
a_{t}:=\inf\left\{ p:\, A_{t}(p)>0\right\} \vee 0,\qquad b_{t}:=\sup\left\{ p:\, B_{t}(p)>0\right\}\wedge n+1, \]
the (best) ask and bid respectively. 
The list of possible events causing jumps of $\Xi$ together with the notation for their intensities is given in the following table:
\bigskip
\def\BMO{\mathrm{BMO}}
\def\SLO{\mathrm{SLO}}
\def\CA{\mathrm{CA}}
\def\SMO{\mathrm{SMO}}
\def\BLO{\mathrm{BLO}}
\def\CB{\mathrm{CB}}
\begin{center}
\begin{tabular}{c|c|l}
code - $e$ &  intensity - $I_t(e)$ & description \\
\hline
$\BMO$ &$\theta_t=\theta(a_t,b_t)$ & 
\begin{minipage}{8cm}
An arrival of a buy market order, causing $A^{a_t}$ to decrease by one (if the sell limit order book is empty then the arrival of the market order has no effect).
\end{minipage}
 \\
 \hline
$\SLO(p)$ &$\kappa_{t,p}=\kappa(a_t,b_t,p)$& 
\begin{minipage}{8cm}
An arrival of a sell limit order with limit price $p > b_t$ causing an increase of $A^p$ by one.
\end{minipage} \\
\hline
$\CA(p)$ & $\rho_{t,p}=A^p_t\rho(a_t,b_t,p)$ & 
\begin{minipage}{8cm}
A cancellation of a pending sell limit order with a limit price $p$ causing a decrease of $A^p$ by one. 
\end{minipage}
\\
\hline
$\SMO$ & $\vartheta_t=\vartheta(a_t,b_t)$ & 
\begin{minipage}{8cm}
An arrival of a sell market order, causing $B^{b_t}$ to decrease by one (if the buy limit order book is empty then the arrival of a market order has no effect).
\end{minipage}
 \\
\hline
$\BLO(p)$ & $\lambda_{t,p}=\lambda(a_t,b_t,p)$& 
\begin{minipage}{8cm}
An arrival of a buy limit order with limit price $p < a_t$ causing an increase of $B^p$ by one.
\end{minipage} \\
\hline
$\CB(p)$ & $\sigma_{t,p}=B^p_t\sigma(a_t,b_t,p)$ & 
\begin{minipage}{8cm}
 A cancellation of a pending buy limit order with a limit price $p$ causing a decrease of $B^p$ by one. 
\end{minipage}
\end{tabular}
\end{center}

\noindent Here, all $\theta,\kappa,\rho,\vartheta,\lambda,\sigma$ are measurable functions. 

It is assumed that 
all the flows of the market orders, and the flows of limit orders as well as their cancellations are mutually independent in the sense that
 the  conditional distribution of the relative time of the first event following $t\in \R^+_0$ given the history of $\Xi$ up to $t$ is exponential with parameter 
\begin{equation}\label{eq:d1}
\Lambda_t=\theta_t+\vartheta_t+ \sum_{p=b_t+1}^n\kappa_{p,t}+
\sum_{p=0}^{a_{t}-1}\lambda_{p,t} +\sum_{p=a_{t}}^nA_t^p\rho_{p,t} +\sum_{p=0}^{b_{t}}B_t^p \sigma_{p,t}
\end{equation} 
and that
\begin{equation}\label{eq:d2}
\text{the probability that the next event will be of type $e$ equals to $I_t(e) / \Lambda_t$}
\end{equation}
where $I_t(e)$ is the intensity of event $e$ at $t$. It is obvious that $\Xi$ is then a Markov chain in a countable state space. 

Finally, it is assumed that $a_0,b_0$ are deterministic and
\begin{multline}
\label{eq:zero}
\text{$A^{a_0+1}_0,A^{a_0+2}_0,\dots$ are Poisson with parameters $\iota_{a_0+1},\iota_{a_0+2}, \dots$,}\\ 
\text{mutually independent and independent of $B_0$}
\end{multline}
where $\iota_\bullet \geq 0$ are constants, and that a symmetric  assumption holds for $B$.

\subsection{Relation to Existing Models}

The following table shows how some of the models mentioned in the Introduction comply with our setting. When speaking about \cite{Luckock03}, we have its discretised version (see \cite{Smid12}, Sec 3.3) in mind. When speaking about \cite{Smith03}, we are considering its bounded version (i.e., contrary to \cite{Smith03} we assume zero arrival intensities for prices smaller than one and greater than $n$
%\footnote{Since, however, the bounds of our %model may be set arbitrarily large, we can %approximate the \cite{Smith03} model by our ZI %one with arbitrarily high accuracy.}
). 
\begin{center}
\begin{small}
\noindent
\begin{tabular}{l|c|c|c|c|c|c}
model & $\theta$ & $\vartheta$ & $\kappa$ & $\rho$ & $\lambda$ & $\sigma$
\\
\hline
\cite{Luckock03} & $K(b_t)$ & $1-L(a_t-1)$ & 
$\kappa^l_p$ & $0$ & $\lambda^l_p$ & $0$
\\
\hline
\cite{Smith03} & $\theta^s$ & $\theta^s$ & 
$\kappa^s$ & $\rho^s$ & $\kappa^s$ & $\rho^s$
\\
\hline
\cite{Cont10} & $\theta^c$ & $\theta^c$ 
& 
$\kappa^c(p-b_t)$ & $\rho^c(p-b_t)$
& 
$\kappa^c(a_t - p)$ & $\rho^c(a_t - p)$
\\
\hline
\end{tabular}
\end{small}
\end{center}
\noindent
Here, $\kappa^l_p=K(p)-K(p-1)$, $\lambda^l_p=L(p)-L(p-1)$ where
$K$, $L$ are (continuous) cumulative distribution functions, $\kappa^c,\rho^c$ are measurable functions and the rest of the  symbols are constants.

Some of the models from the Introduction were not mentioned in the table: We did not include either \cite{Maslov00} or \cite{Challet01} because they both consider discrete time and are very similar to \cite{Smith03} with $\rho=\sigma=0$, \cite{Cont10}, respectively. The model by \cite{Mike08} was not included because of its complicated cancellation sub-model and because, apart from the cancellations, it is similar to that of \cite{Cont10}. Finally, we did not include \cite{Stigler64} because it is a special version of \cite{Luckock03} (with $K(x)=L(x)=x$).

\subsection{Distribution}

Let us start with a result which will guarantee that sampling from the model will give enough information to a statistician.

\begin{proposition}\label{ergprop} If  $\rho(\bullet)> 0$, $\sigma(\bullet)>0$, $\theta(\bullet) > 0$ and $\vartheta(\bullet)>0$. then 
$\Xi_t$ is ergodic. 
\end{proposition}
\begin{proof} The Proposition may be proved analogously to \cite{Cont10}, Proposition 2, where the ergodicity is verified by finding a Markov chain in $\N_0$ dominating the total number of orders with a recurrent zero state. In particular, it suffices to replace the definition of $\lambda$ and $\mu_i$ from \cite{Cont10} by 
$$\lambda= \sum_{p=1}^n [\max_{a,b,p} \kappa(a,b,p)+
\max_{a,b,p} \lambda(a,b,p)],
$$
and
$$
\mu_i = \min_{a,b} \theta(a,b)+\min_{a,b} \vartheta(a,b) 
+ i\min_{a,b,p}[\min(\rho(a,b,p),\sigma(a,b,p)].
$$
\end{proof}

\noindent Our next goal is to derive a recursive analytic expression of the distribution of 
$$
x_\tau = (a_\tau,b_\tau,q_\tau,r_\tau)_{\tau \geq 0}, 
\qquad 
q_\tau = A^{a_\tau}_\tau,
\quad
r_\tau = B^{b_\tau}_\tau,
\quad
\tau \geq 0
$$
(the L1 process). To do so, let us denote by $t_1,t_2,\dots$ the jumps of $x$ and, in order to avoid frequent double indexing, let us write $y_{[i]}$ instead of $y_{t_i}$ and $y_{[i-]}$ instead of $y_{t_{i}-}$ for any process $y$.
For each $i$, we want to evaluate
$$
\P[(\Delta t_i, x_{[i]}) \in \bullet|\xi_{[i-1]}],\qquad
\xi_{[k]}=(x_{[k]},\Delta t_k,x_{[k-1]},\dots,x_0),
\quad k\geq 0,
$$
starting with 
$$
\P[(\Delta t_i, a_{[i]}, q_{[i]}) \in \bullet|\xi_{[i-1]}].
$$
To this end, note that 
\begin{itemize}
\item $a$ jumps down if and only if a limit order arrives into the spread, in which case a limit price of the new order becomes a new value of $a$.
\item $a$ jumps up if and only if the offered volume of the ask decreases to zero due to either a market order arrival or a cancellation, in which case the closest occupied tick becomes a new value of $a$.
\end{itemize}
Formally,
\def\DA{\mathrm{DA}}
\begin{equation}\label{eq:aq}
(a_{[i]},q_{[i]}) = 
\begin{cases}
(b_{[i-1]} + 1,1)  &  e_i = \SLO(b_{[i-1]}+1) \\
(b_{[i-1]} + 2,1)  &  e_i = \SLO(b_{[i-1]}+2) \\
\dots & \\
(a_{[i-1]} - 1,1)  &  e_i =  \SLO(a_{[i-1]}-1) \\
(a_{[i-1]},q_{[i-1]}+1)  &  e_i =  \SLO(a_{[i-1]}) \\
(a_{[i-1]},q_{[i-1]}-1) & e_i \in \{\BMO,\CA(a_{[i-1]}) \}, q_{[i-1]}>1 \\
(d_i,A^{d_i}_{[i-]}) & e_i \in \{\BMO,\CA(a_{[i-1]}) \}, q_{[i-1]}=1 \\
(a_{[i-1]},q_{[i-1]}) & \text{otherwise}
\end{cases}
\end{equation}
$$
d_i = \inf\{\pi: \pi > a_{[i-1]}, A^{\pi}_{[i-]}>0\}
\wedge n+1
$$
where $e_i$ is the type of an event happening at $t_i$ and
where $A^{n+1}\equiv 0$ by definition.
Thus, to determine the conditional distribution of $(\Delta t_i,a_{[i]},q_{[i]})$, it suffices to know a joint distribution of $(\Delta t_i, e_i,A_{[i-]})$ which is described by the following Theorem in three steps:

\begin{theorem}\label{th:dist} 
(i) For any $\tau$ and $e$,
$$
\P[\Delta t_i > \tau,e_i=e|\xi_{[i-1]}]
=\exp\{-\gamma_{i-1} \tau\} \pi_{i} (e),
$$
where
$$
\pi_{i}(e)=  \gamma^{-1}_{i-1} \times
\begin{cases}
\kappa_{p,[i-1]}, & e = \SLO(p), b_{[i-1]} < p \leq a_{[i-1]},\\
\theta_{[i-1]}, & e= \BMO, \\
\rho_{[i-1]}, & e= \CA(a_{[i-1]}), \\
\lambda_{p,[i-1]}, & e = \BLO(p), b_{[i-1]} \leq p < a_{[i-1]},\\
\vartheta_{[i-1]}, & e= \SMO, \\
\sigma_{[i-1]} & e= \CB(b_{[i-1]}),
\end{cases}
$$
and, for any $k$, 
%$$
%\theta_k=\theta(a_{[k]},b_{[k]}),
%\qquad
%\vartheta_k=\vartheta(a_{[k]},b_{[k]}),
%\qquad
%\rho_k = q_{[k]}\rho(a_{[k]},b_{[k]},a_{[k]}), 
%\qquad
%\sigma_k = r_{[k]}\sigma(a_{[k]},b_{[k]},a_{[k]}),
%$$
$$
%\kappa_{p,k} = \kappa(a_{[k]},b_{[k]},p),
%\qquad
%\lambda_{p,k} = \lambda(a_{[k]},b_{[k]},p),
%\qquad
\gamma_k = 
\theta_{[k]}+\vartheta_{[k]}
+ \rho_{[k]}+\sigma_{[k]}
+
\sum_{p=b_{[k]}+1}^{a_{[k]}} \kappa_{p,[k]}
+\sum_{p=b_{[k]}}^{a_{[k]}-1} \lambda_{p,[k]}.
$$
(ii) Denote by $\dirac_c$ the  Dirac measure concentrated in $c$ (i.e., the distribution of constant $c$) and write $\circ$ for convolution (i.e., summation of two independent variables). For any $i$ and $p$,
$$
A^p_{[i-]}|\xi_{[i-]}\sim 
\begin{cases}
0 & \text{on $[1 \leq p < a_{[i-1]}]$}\\
\dirac_{q_{[i-1]}} & \text{on $[p = a_{[i-1]}]$} \\
\mathrm{Bi}
\left(
\nu_{p,i},\varpi_{p,i}
\right) \circ
\mathrm{Po}
\left(\epsilon_{p,i}+\iota_{p,i}\right) & \text{on $[a_{[i-1]} < p \leq n]$}
\end{cases}
$$
where 
$$
\xi_{[i-]}=(e_i,\Delta t_i, \xi_{[i-1]}),
\qquad
\nu_{p,i}=
\begin{cases}
A^p_{[k_i(p)-]}& k_i(p) > 0, \\
0 & k_i(p)=0,
\end{cases}
\qquad
\varpi_{p,i}
=
\begin{cases}
\prod_{j=k_i(p)}^{i-1}\delta_{p,j} & k_i(p) < i \\
1 & k_i(p)=i
\end{cases}
$$
$$
\epsilon_{p,i}=\sum_{j=k_i(p)}^{i-1}\phi_{p,j}(1-\delta_{p,j})\prod_{m=j+1}^{i-1}\delta_{p,m}, 
\qquad
\iota_{p,i}=\begin{cases}
0 & k_i(p) > 0, \\
\iota_p\prod_{j=0}^{i-1}\delta_{p,j} &  k_i(p) = 0,
\end{cases}
$$
$$
k_i(p) = \min \{k \geq 0,a_{[k]} < p \} \vee i
$$
and, for any $k$,
$$
\phi_{p,k}=\frac{\kappa_{p,[k]}}{\rho_{p,[k]}},
\qquad
\delta_{p,k} =\exp\{
-\rho_{p,[k]}\Delta t_{k+1}
\}
$$
(iii) 
$A^1_{[i-]},A_{[i-]}^2,\dots,A_{[i-]}^n$ are mutually conditionally independent given $\xi_{[i-]}$.
\end{theorem}

\begin{proof} See Appendix \ref{sec:proofs}.
\end{proof}

\noindent Before proceeding further, let us illustrate point (ii) of the Theorem, which is somewhat opaque, by an example.

\begin{xmpl} \label{xmpl:x} Consider model by \cite{Smith03} with $\iota_\bullet\equiv 0$ (implying $\iota_{\bullet,3}\equiv 0$). 
Let $i=3$, $a_0=2$, $q_0=1$,
$a_{[1]}=3$, $q_{[1]}=2$, $a_{[2]}=1$, $q_{[2]}=3$ (see Figure \ref{fig:ilust} for an illustration). 

Let us start with $p=2$: as no orders could be present in tick $2$ before $t_2$ (the ask was above 2 that time), all the in-book orders present in the tick $t_3$ are those having arrived between $t_2$ and $t_3$ and not being cancelled. According to Appendix \ref{sec:proofs}, number of those orders is $\Po(\epsilon_{i,3})$ (to check it, note that  $k_i(3)=2$ so $\nu_{i,2}=A^2_{[2-]}=0$ and $\epsilon_{i,3}=\frac{\kappa}{\rho}[1-\exp\{-\rho \Delta t_3\}]$ which, according to Appendix \ref{sec:proofs}, is actually the distribution of the immigration and death process of length $\Delta t_3$). 

For $p=3$, the situation is the same with the difference that number of initial orders having survived until $t_3$, which is distributed according to $\Bi(2,\exp\{-\rho\Delta t_3\})=\Bi(A^3_{[3-]},\delta_{i,3})$, has to be added to $A^p_{[i-]}$.

Further, as the orders present in tick 4 at $t_3$ are exactly those having arrived  from the start and not being cancelled, $A^4_{[i-]}|\xi_{[i-]}\sim \Po\left(\frac{\kappa}{\rho}[1-\exp\{-\rho  t_3\}\right)$ - it is easy to check that (ii) gives the same result.
 
Finally, because $k_i(1)=3$, we have  $\epsilon_{1,i}=0$, $\omega_{1,i}=1$ and $\nu_{1,i}=3$ so, according to (ii), $A^1_{[i-]}\sim \Bi(3,1)=\dirac_3$.
\end{xmpl}

\begin{figure}
\begin{center}
\resizebox*{10cm}{!}{\includegraphics{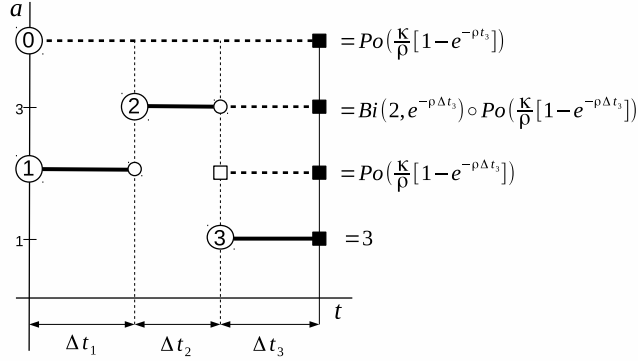}}
\end{center}
\caption{Illustration of Example \ref{xmpl:x} }
\label{fig:ilust}
\end{figure}

\noindent From the Theorem and (\ref{eq:aq}), we get that

\begin{corollary} \label{cor:aq} 
(i) Conditional density of $(e_i,\Delta t_i|\xi_{[i-1]})$ is given by $$
f_i(\tau,e|\xi_{[i-1]}) = \gamma_{i-1}^{-1}\exp\{-\gamma_{i-1} \tau \} \pi_i(e),
$$
(ii) conditional density of $(a_{[i]},q_{[i]})|\xi_{[i-]}$ is given by 
\begin{multline*}
g_i(a,q|\xi_{[i-]})\\=
\begin{cases}
\mathbf{1}[q= A^a_{[i-1]}+1] & 
\text{on $[b_{[i-1]} < a \leq a_{[i-1]},
e_i=\SLO(a)]$} \\
\mathbf{1}[q= q_{[i-1]}-1] & 
\text{on $[a = a_{[i-1]},
e_i\in\{\CA(a_{[i-1]}),\BMO\}, q_{[i-1]} > 1]$} \\
\omega_{a,q,i}
\prod\limits_{\pi=a_{[i-1]}+1}^{a-1} 
(1-\varpi_{\pi,i})^{\nu_{\pi,i}} e^{-\epsilon_{\pi,i}-\iota_{\pi,i}}
& 
\text{on $[a_{[i-1]}<a,e_i\in\{\CA(a_{[i-1]}),\BMO\}, q_{[i-1]} = 1]$} \\
\mathbf{1}[q=q_{[i-1]},a=a_{[i-1]}] & \text{otherwise},
\end{cases}
\end{multline*}
where $\mathbf{1}$ denotes indicator function and
$$
\omega_{p,q,i} = 
\begin{cases}
\P[\Bi(\nu_{p,i},\varpi_{p,i}) \circ \Po(\epsilon_{p,i}+\iota_{p,i}) = q] & 
a_{[i-1]} < p\leq n,\\
\mathbf{1}[q=0] & a=n+1.
\end{cases}
$$
\\
(iii) conditional density of $(a_{[i]},q_{[i]},e_i,\Delta t_i)$ given $\xi_{[i-1]}$ is given by
$$
h_i(a,q,e,t|\xi_{[i-1]}) = g_i(a,q|e,t,\xi_{[i-1]})f_i(e,t|\xi_{[i-1]}).
$$
\end{corollary}

\begin{remark} As the definitions of $B$ and $A$ are symmetric, formulas for the (conditional) distribution of $(b,r)$ are symmetric to those for $(a,q)$. Moreover, as changes of $(a,q)$ and $(b,r)$ are mutually exclusive almost sure, i.e., 
$$
[\text{$(a,q)$ changes at $t$}] \Rightarrow 
[\text{$(b,r)$ does not change at $t$}]
$$
almost sure and vice versa, $(b,r)$ is conditionally constant given events which cause changes of $(a,q)$ (and vice versa) so the conditional distribution of $x_{[i]}|\xi_{[i-]}$ is uniquely determined by Corollary \ref{cor:aq} and a formula for $(b,r)$. 
\end{remark}

\noindent In the rest of the paper, we shall deal only with the ask side, i.e., with $(A,a,q)$, leaving $(B,b,q)$ aside due to the symmetry.

Because density $h_i$ depends on  all the parameters related to $A$ (i.e.,  $\kappa,\rho,\theta,\iota$), it is straightforward to estimate these parameters by Maximum Likelihood based on $h_i$ and a sample from ($\Delta t_{[i]}, e_{[i]}, a_{[i]}, q_{[i]}$). Moreover, if we put 
\begin{equation}
\label{eq:iotazero}
\iota \equiv 0
\end{equation}
and, quite realistically, assume that $\kappa$ and $\rho$ depend only on a distance to the ask, i.e., 
\begin{equation}
\label{eq:relkr}
\kappa(a,b,p)=\tilde \kappa(a-p), \qquad \rho(a,b,p)=\tilde \rho(a-p)
\end{equation}
for some $\tilde \kappa$ and $\tilde \rho$, and that the ``inside-the-book'' intensities
\begin{equation}
\label{eq:parsob}
\tilde \kappa|_{\{1,2,\dots\}}, 
\quad
\tilde \rho|_{\{1,2,\dots\}},
\end{equation}
depend on different parameters than the ``spread'' intensities
\begin{equation}
\label{eq:parsspr}
\theta,
\quad
\tilde \kappa|_{\{0,-1,-2,\dots\}},
\quad 
\tilde \rho|_{\{0,-1,-2,\dots\}},
\end{equation}
then the estimation of  the ``ask-side'' parameters could be split into the estimation of  the parameters underlying (\ref{eq:parsspr}), based on $f_i(\Delta t_i,e_i)_{i\in \N}$,  and that of the parameters underlying (\ref{eq:parssob}), based on $g_i(a_{[i]},q_{[i]})_{i\in \N}$.
%\begin{equation}
%\label{eq:sample}
%(a_{[i]},q_{[i]}|a_{[i]}>0)_{i\in \N}
%\end{equation}
%(the observations with $a_{[i]} \leq 0$ might be neglected because they %are conditionally constant hence independent of any parameter). 
In the present paper, only the latter case of an estimation is discussed. 
%A proof that the ML estimator based on $g$ is, %despite 
%heterogeneity and dependence of sample %$(a_{[i]},q_{[i]})_{i\in \N}$, consistent and %asymptotically normal is given in Appendix %\ref{sec:mle}.

\section{Generalised Model}
\label{sec:gzi}

Despite the popularity of the ZI models, there is no doubt that they are too rough and neglect many aspects of real-life trading. The most obvious issues in this respect are two: 
\begin{enumerate}
\item[(i)] the volumes are non-unit in reality,
\item[(ii)] agents, at least sometimes, act strategically rather than in a random way.
\end{enumerate}
To deal with (i), authors of ZI models usually argue that the sizes of a majority of orders are one-lot and, when an order is larger, it may be imagined that several subsequent one-lot orders have been issued. However, this simplification may be tolerable only in the case of limit orders (if the order book is observed only at certain time instants then larger orders are indeed interchangeable with several subsequent unit ones) but it is unacceptable for non-unit market orders which could not be imaginatively split without a serious violation of the assumption that their arrival intensity is constant.

To justify the ZI models against (ii), it is usually argued that the strategic reasoning is so complex that taking it as random is less evil than constructing wrong models. This argument, however, has a drawback, too: no matter how reasonable it may sound, one still has to deal with empirical phenomena stemming from (bounded) rationality, such as shifts and rapid insertions and cancellations of orders (especially quotes) in response to changes of the book.

To meet those issues, the setting defined in Section \ref{ss:zi} might be generalised by assuming tbat
\begin{description}
\item[M] volumes of market orders are possibly non-unit,
%; once an order with volume $s$ arrives, the %order books are changed as if $s$ market orders %were issued

\item[S1] once a (unit) limit order stops to be a quote due to an in-spread sell limit order, it is, with probability $1-\eta$, immediately cancelled or shifted to the position of the new quote,
\item[S2] once a quote jumps out of the spread, there is a non-zero probability that the move was caused not by a market order or cancellation but by a shift of the quote.
\end{description}
\def\SAL{\mathrm{SAL}}
\def\SAR{\mathrm{SAR}}
Moreover, to give more freedom to modelling of (possibly complex) behaviour of the quotes, it is allowed that 
\begin{description}
\item[D] the distribution of $\Delta t_i,e_{i}|\xi_{[i-1]}$ is arbitrary and, moreover, $q$ and $r$ may jump by more than one unit.
\end{description}
The list of events potentially changing $x$ at $t_i$ is newly as follows:
\begin{center}
\begin{tabular}{c|l}
code &   description \\
\hline
$\BMO(z)$ & a buy market order of volume $z$ \\
$\SLO(p,z)$ &
a sell limit order with price 
$b_{[i-1]} < p \leq a_{[i-1]}$ and
size $z$ put into the spread  \\
$\CA(a_{[i-1]},z)$ & a cancellation of $z$ orders of the ask \\
$\SAL(z,a)$] & a shift of $z$ orders with price $a_{[i-1]}$  to tick $a_{[i-1]}-a$, $a > 0$ 
\\
$\SAR(z,a)$] & a shift of $z$ orders with price $a_{[i-1]}$  to tick $a_{[i-1]}+a$, $a > 0$ 
\end{tabular}
\end{center}
plus symmetric events concerning $b$.

Even though, given this generalisation, (i) of Theorem \ref{th:dist} ceases to be true, 
its Assertions (ii) and (iii) would keep holding with $\eta \varpi_{p,i}$ instead of $\varpi_{p,i}$ if it is additionally assumed that
\begin{itemize}
\item the volumes of market orders are i.i.d. random, independent of all the past and the present events on the market
\item  the (Bernoulli) variables indicating the shifts are mutually independent, independent of the past and the present events on the market.\footnote{To see it, note that the proof continues to be valid with $\P(\Bi(q_{[i-1]},\eta)=\alpha)$ instead of $\mathbf{1}[\alpha = q_{[i-1]}]$ in (\ref{eq:fixed}).}
\end{itemize}
What does change, however, is the distribution of quote jumps outside of the spread after the ask is depleted, which are now non-zero if and only if
\begin{equation}
\label{eq:eset}
[e_i \in \{\BMO(z),\CA(q_{[i-1]}),\SAR(z,a)\},z \geq q_{[i-1]}, a> 0],
\end{equation}
in which case
$$
(a_{[i]},q_{[i]}) = (a,q)
\Leftrightarrow
\begin{cases}
M_{a,i} \leq s_i, M_{a,i} + A^a_{[i-]} = s_i + q,
& 
a \leq n, \\
M_{n+1,i} \leq s_i & a = n+1, q=0,
\end{cases}
$$
where
$$
M_{a,i}= \sum_{j=a_{[i-1]}+1}^{a-1} A^j_{[i-]}
$$
is the total number of the orders in the book between the former and the present value of the ask and 
$$
s_i
=
\begin{cases}
z - q_{[i-1]} & \text{if $e_i=\BMO(z)$,} \\
-z & \text{if $e_i=\SAR(z,a)$ for some $a$,}\\
0 & \text{if $e_i = \CA(a_{[i-1]}).$}
\end{cases}
$$
is the number of the orders the particular change would like to ``eat'' from the inside of the book. 

As, by (ii) and (iii) of (the modified version of) Theorem \ref{th:dist},
$$ 
M_{a,i}|\xi_{[i-]} \sim 
\Bi(\nu_{a_{[i-1]}+1,i}, \eta \varpi_{a_{i-1}+1,i}) 
\circ
\dots
\circ
\Bi(\nu_{a-1,i}, \eta \varpi_{a-1,i}) 
\circ \Po
\left(
\sum_{j=a_{[i-1]}+1}^{a-1} [\epsilon_{j,i}+\iota_{j,i}]
\right)
$$
and, by (iii) of the Theorem, $A^a_{[i-]}$ is conditionally independent of $M_{a,i}$ given $\xi_{[i-]}$, we are getting

\begin{corollary} \label{cor:aqh} On set (\ref{eq:eset}),
the distribution of $(a_{[i]},q_{[i]})|\xi_{[i-]}$ is given by density
$$
\tilde g_i(a,q|\xi_{[i-]})=
\begin{cases}
\sum_{j=0}^{s_i} \P[M_i=j|\xi_{[i-]}]
\P[A^a_{[i-]}= s_i+q-j|\xi_{[i-]}]
& 
a_{[i-1]}<a \leq n,
\\
\P[M_i\leq s_i|\xi_{[i-]}]
& 
a = n+1, q=0.
\end{cases}
$$             
\end{corollary}

\begin{remark} As both the value of the ask and its volume are uniquely determined by $\xi_{[i-]}$ outside set (\ref{eq:eset}), the distribution of $(a_{[i]},q_{[i]})|\xi_{[i-]}$ is Dirac outside set (\ref{eq:eset}).
\end{remark} 

\noindent Thus, if we keep assuming (\ref{eq:iotazero}) and (\ref{eq:relkr}), we may estimate parameters underlying $\tilde \kappa|_{\{1,2,\dots\}}$ and $\tilde \rho|_{\{1,2,\dots\}}$ together with $\eta$ by means of MLE based on $\tilde g_i$; as it could be checked in Appendix \ref{sec:mle}, the asymptotic properties of the estimator are not harmed by the generalisation provided that the generalised process is ergodic, which may be guaranteed, e.g., by requiring that
\begin{description}
\item[E] $q_{t}+r_t$ is stochastically dominated by a Markov chain with a zero recurrent state.
\end{description}
Under this assumption, Proposition \ref{ergprop} keeps holding because none of the generalisations, except for D, which is treated by E, increase the total number of orders in $\Xi$ in comparison with the ZI.

\begin{remark} Unlike in the unit-volume ZI model, price impact may be predicted in the GZI model: denoting $m_i$ the impact of a trade with volume $z$ at time $t$, we have
\begin{multline*}
\P[m_i(z)= m|\xi_{[i-1]}=\xi,t_i = t,, e_i = \mathrm{BMO}(z)]
=\P[\Delta a_{[i]}= m|\xi_{[i-1]}=\xi,t_i = t, e_i = \BMO(z)]
\\
=\sum_{q=0}^\infty \tilde g_i(a_{[i-1]}+m,q|\xi_{[i-1]}=\xi,t_i = t, e_i = \BMO(z)).
\end{multline*} 
\end{remark}

\section{Empirical Evidence}
\label{sec:empirical}

\subsection{Data}
\label{datasec}
For our empirical analysis, tick-by-tick trade and quote data was used, i.e., values
$$
(x_{[i]},\tau_{[i]})_{i\in \N}
$$
where $\tau_{[i]}$ is an amount traded at $t_i$, defined by
$$
\tau_{[i]}=
\begin{cases}
z & e_i = \SMO(z) \\
-z & e_i = \BMO(z) \\
0 & \text{otherwise}.
\end{cases}
$$
In particular, data of six stocks
\begin{itemize}
\item Exxon Mobile (XOM) 
\item Microsoft (MSFT), 
\item General Electric (GE) 
\item MarketAxess Holdings (MKTX)
\item J2 Global (JCOM)
\item American Realty Investors (ARL)
\end{itemize}
from nine US electronic markets 
\begin{itemize}
\item NASDAQ OMX BX
\item NSE
\item Chicago
\item NYSE
\item ARCA
\item NASDAQ T
\item CBOE
\item BATS
\item ISE
\end{itemize} 
from ten months starting from March $2009$ were analysed.

Our choice of the titles was done so as to cover both ``large'' and the ``smaller'' stocks. The first three ones - XOM, MSFT and GE -  belong to the top ten companies by market capitalization and usually exhibit large trading volumes. MKTX and JCOM, on the ather hand, are usually ranked as ``small caps'' (having market capitalization around \$1 bilion) with moderate trading volumes. Company ARL, belonging to the category of ``micro caps'' (capitalisation around \$100 million), could be counted among illiquid stocks.

As to the markets, we did not pre-select which of them to analyse; contrarily, all the markets were included on which the examined stocks were traded and for which at least some data were available\footnote{Our dataset came from Tickdata, Inc.}- the only exception is NASDAQ ADF, which is a platform for recording trades rather than a limit order market (see \cite{FinraADF}).

As the trade data (process $\tau$) came from a different data source than the quote data (process $x$), it was necessary to algorithmically match individual trades with the L1 data changes in which our algorithm, designed for that purpose, succeeded in about 70 percent of trades (the rest could not be uniquely attributed to any L1 change ). This fact, however, did not harm the estimation other way than by decreasing the sample size.\footnote{To be absolutely rigorous, we should assume that the success in matching is stochastically independent of the quote process  to claim this.}

\subsection{Estimation}

For the actual estimation, only records originating between 9:40 a.m. and 3:30 p.m. were used when the process could be assumed to be near its stationary distribution; the inclusion of the ten-minute ``warm-up'' period following the opening at 9:30 also partially justifies our assumption (\ref{eq:iotazero}).

When estimating within the ZI models, all the jumps of $a$ up\footnote{The jumps of $a$ down and the changes of $(x,\tau)$ preserving $a$ could be omitted without any loss of information because the corresponding these changes are conditionally constant hence not depending on any parameter.} could serve as a sample while, in the case of the GZI models, only the jumps of $a$ matched with trades were included; the reason for this restriction is that the jumps of $a$ caused by unpaired trades, shifts and cancellations could not be distinguished from each other in the L1 data, so the corresponding values of $s_i$ could not be determined. For each paired trade record, on the other hand, the value of $s_i$ may be obtained from the volume of the trade and from the (known) volume of the former ask.

% See Figure \ref{summary} for a first look to the data

%\begin{figure}
%\input{summary}
%%\input
%\caption{Summary data for the examined stocks and markets. The table: %$m/y$ - month/year, $\bar s$ - average spread ($s_t = a_t-b_t$), %$\overline{\Delta t_i}$ - average time between jumps of quotes %$(a_t,b_t)$, $\#$ - number of jumps of the quotes per day (in 10,000). %The graph: vertical lines - numbers of the quotes' jumps, the curve - %average spread.}
%\label{summary}
%\end{figure}

For both the ZI and the GZI settings, and for each stock-market pair, the following estimation procedure was run: as its first step, a simple version of the model with 
\begin{equation*}
\tilde\kappa(\bullet)\equiv \kappa_0, \quad \tilde \rho(\bullet)\equiv \rho_0, \quad \rho_0>0,\quad  \kappa_0>0\qquad (S)
\end{equation*} (\cite{Smith03} in its ZI variant), which we later call {\em basic model}, was estimated. Subsequently, three {\em power tail} models with $\tilde \kappa$ and $\tilde \rho$, defined by 
\begin{equation*}
\tilde \kappa(i)=\begin{cases} 
\kappa_{i-1} & i \leq n \\
\kappa_{n-1} (i-n)^{\alpha_\kappa} & i > n
\end{cases}
\quad
\tilde \rho(i)=\begin{cases} 
\rho_{i-1} & i \leq n \\
\rho_{n-1} (i-n)^{\alpha_\rho} & i > n,
\end{cases}
 \qquad(T_n) \qquad \qquad n=1,2,3,
\end{equation*}
were estimated gradually. As a ``true'' model, the one was chosen in which all the parameters came out significant while the likelihood ratio comparing this model with the subsequent one came out insignificant. 

Each estimation procedure was run on a sample of at most $5,000$ observations; even though there were many more observations available for a majority of the stock-market pairs, a restriction had to be made due to the large time requirement of the estimation, caused by complex formulas for conditional densities involved. In case that the optimisation algorithm could not find maximum of the ML function within a time limit of $2,500$ sec (approx. $42$ min) for none of the four variants of the model, the procedure was repeated once more with a sample of size $1,000$.

The prediction power of each model $m$ was evaluated by 
$$
P_m=1-
\frac
{\sum_{i=N}^{N+M}|a_{[i]}^+- \E(a_{[i]}^+|\xi_{[i-]})|}
{\sum_{i=N}^{N+M}|a_{[i]}^+ - \bar {a^+}_N|}
$$ 
where $N$ is a sample size,  $M\doteq N/10$ is the number of out-of-sample observations, $a^+_1,\dots,a^+_{M+N}$ are the magnitudes of the jumps of $a$ up, $\bar a^+_N$ is their average and where the conditional expectation in the numerator was computed by means of the estimated parameters. By its construction, $P_m$ can be understood as a percentage improvement in comparison with a naive prediction of the out-of-spread jumps by their mean. In the GZI setting, $P_m$ may also be seen as a measure of accuracy of the market impact prediction.

For automation of the whole procedure including matching of trades, and the model selection, a {\tt C++} package has been developed by the author, using the NLOPT library, namely LBFGS and BOBYQUA optimisation algorithms, to compute maxima of the ML functions.\footnote{As a default choice LBFGS  has been used, the BOBYQUA was employed only in instances in which it could not find a maximum within the time limit). For more on these algorithms and the NLOPT library, see \cite{Johnson12}.} A source code of the package is available at \url{https://github.com/cyberklezmer/fepp} under branch {\tt qf15}.

\subsection{Results}

Results of the estimation procedure of the ZI modelsd for the ``big-cap'' stocks XOM, MSFT and GE are summarised in Table \ref{table:sumz}. It may be seen that 21 out of all the 27 estimations were successful in the sense that a significant model was found, which exhibits a positive improvement (measured by $P_m$). Out of the ``failed'' cases, four times a significant model was found giving a negative $P_m$\footnote{In two of these instances, all the jump increments were exactly one; hence their average predicted the jumps with zero error, which lead to minus infinity for $P_m$.}, once the time limit was reached even with a reduced sample, and once no data was available.

Out of all the 21 successful instances, the basic model $(S)$ was chosen nine times, the simplest tail model $(T_1)$ six times and model $(T_2)$ six times; model $(T_3)$ never came out significant within the time limit. 
Higher $P_m$'s values in comparison with a basic model were reached in only seven out of the 12 cases when a tail model was selected, while the basic model came out better in two cases; in the three remaining cases, the $P_m$ was identical up to two decimal digits. Even though, on average, the tail model performed better, the difference was not found statistically significant according to a Wilcoxon test. Detailed results of this estimation can be found in Appendix \ref{zires}.

Summary results for the GZI model for big-caps may be found in Table \ref{table:sumg}. Here, the estimation was successful in 21 out of 25 cases with at least some data, three times the time-out was reached;  $P_m$ came out negative in one case.

In order to answer the question whether it is worth using the more complicated GZI model rather than the ZI one, the estimation procedure was run for ZI models once more using the same samples as those used for the GZI ones (i.e., only the changes with trades matched); the corresponding results are displayed in the second column below each stock. Out of the 20 cases when both estimations were successful, eight times the GZI model reached a higher $P_m$ while three times $P_m$ was higher for the ZI model; nine times $P_m$ was the same up to two decimal digits. Again, it could not be statistically proved that the GZI model performs significantly better.

Detailed results of the GZI estimation may be found in Appendix \ref{gzires} (the GZI models) and Appendix \ref{gzitres} (their ZI counterparts).

Tables \ref{table:alphakappa} and \ref{table:alpharho} summarise estimates of tail exponents of $\tilde \kappa$ and $\tilde \rho$. It may be seen that a dispersion of the values of $\alpha_\kappa$ is great; however, their average $-2.3(0.41)$ is not far from present empirical evidence (see, e.g., \cite{chakraborti2011econophysicsi}, III.C).\footnote{It should be noted, however, that only the first several values of $\tilde \kappa$ and $\tilde \rho$ played a role in the estimation, especially when the average jump out of the spread was small, so the results say very little about the actual power law behaviour.} The results for tail exponents of $\tilde \rho$, which have never been statistically estimated yet to the best of the author's knowledge, are similar to the case of $\alpha_\kappa$, with average value $-2.48(0.42)$.

Results of the estimation of the small caps MKTX, JCOM and the micro-cap ARL may be found in Table \ref{table:sumzs} and Table \ref{table:sumgs}. For the small caps, the ZI model was successful for 10 of 13 pairs with sample size at least $N>20$, once a negative $P_m$ was reached, two times the optimization failed.\footnote{In particular, the solver stopped the optimization without changing the initial parameters reporting that its tolerance criterion was reached which suggest that the MLE function is too flat to be optimized.} For ARL, the procedure was successful only once (with a disappointing prediction power) with all of the failures due to failed optimizations. A closer look to the results suggests two possible causes of the failures: small sample sizes and/or large average jumps of $a$, causing large evaluation times of the densities (see Appendix \ref{sec:detail} for detailed values).
The situation is not much better in the case of GZI models either: small caps are successful in 7 out of 9 cases, ARL is only once half-way successful out of three cases. 
Results of a comparison of more complex ZI models with their basic variant are similar to those of the big-caps: out of 8 cases when the comparison is possible, more complex models won three times, four times the basic model was more successful, once the results were the same. Similarly, from 8 comparisons of GZI and ZI variants of the model, the GZI variant was more successful four times, the ZI one three times, once the results were equal.

\begin{table}
\begin{center}
\begin{tabular}{l|C{2.7cm}|C{2.7cm}|C{2.7cm}}
& XOM & MSFT & GE \\
\hline
NASDAQ OB & $T_1(<0)$     
		  & $T_1(0.70,0.70)$  
		  & $S(0.48)$  \\
NSE       & $T_2(0.91,0.90)$ 
          & $T_2(0.89,0.87)$ 
          & $T_1(0.76,0.76)$  \\
Chicago   & $T_2(0.47,n/a)^-$  
          & $S(0.93)$ 
          & $T_1(0.76,0.71)$ \\
NYSE      & $S(0.33)$ 
          & $\times$ 
          & $S(0.53)$ \\
ARCA      & $T_1(0.92,0.91)$ 
          & $S(-\infty)^\star$          
          & $t$   \\
NASDAQ T  & $T_1(0.53,0.53)$            
          & $S(1.00)^\star$  
          & $S(-\infty)^\star$ \\ 
CBOE      & $S(<0)$ 
          & $T_2(0.35,0.36)$ 
          & $S(0.67)$ \\
BATS      & $T_2(0.86,0.83)$            
          & $S(1.00)$ 
          & $S(0.00)^\star$ \\
ISE       & $T_2(0.66,0.67)$ 
          & $S(0.66)$ 
          & $T_1(0.94,0.93)^\star$ \\
\end{tabular}

\caption{Results for ZI. Notation: $m(p,p_0)$ - model $m$ was selected with $P_m=p$ and with $P_S=p_0$, minus sign - sample less than $5000$ was available, $\star$ - a reduced sample used, $t$ -  time limit reached when without finding a result, $\times$ -  less than $20$ observations were available. }
\label{table:sumz}
\end{center}
\end{table}

\begin{table}
\begin{center}

\begin{tabular}{l|C{1.75cm}C{1.75cm}|C{1.75cm}C{1.75cm}|C{1.75cm}C{1.75cm}}
&\multicolumn{2}{c}{XOM} & \multicolumn{2}{c}{MSFT} & 
\multicolumn{2}{c}{GE} \\
& GZI & ZI & GZI & ZI& GZI & ZI \\
\hline
NASDAQ OB & $T_2(0.19)$  & $T_2(0.17)$     
		  & $T_1(0.68)$  & $T_1(0.66)$
		  & $T_1(0.74)$  & $S(0.74)$\\
NSE       & $T_1(0.70)^\star$        & $T_1(0.69)^\star$ 
          & $T_2(0.38)$  & $T_2(0.29)$
          & $S(0.33)$ & $t$\\
Chicago   & $T_2(0.35)^-$& $T_2(0.34)^-$
          & $T_2(0.49)^-$& $T_1(0.51)$    
          & $T_1(0.79)$  & $T_1(0.79)$\\
NYSE      & $S(0.17)$  & $S(0.13)$  
          & $\times$        & $\times$
          & $S(1.00)^\star$        & $S(1.00)^\star$\\
ARCA      & $T_2(0.84)$  & $T_2(0.84)$        
          & $t$        & $S(0.00)^\star$
          & $t$        & $t$\\
NASDAQ T  & $S(0.74)$  & $T_1(0.74)$        
          & $\times$        & $\times$    
          & $t$        & $t$\\ 
CBOE      & $T_2(<0)^-$   & $T_2(<0)^-$      
          & $T_2(0.72)^-$ & $S(0.68)^-$
          & $T_2(0.23)$  & $T_2(0.23)$\\
BATS      & $T_1(0.31)$  & $T_1(0.60)$         
          & $S(0.00)^\star$        & $S(0.00)^\star$
          & $S(0.00)^\star$    & $S(0.00)^\star$\\
ISE       & $T_1(0.44)^\star$        & $S(0.42)^\star$
          & $T_1(0.96)^\star$        & $T_1(0.96)^\star$
          & $S(0.36)^\star$  & $S(0.42)^\star$
\end{tabular}

\caption{Results for GZI. Notation: $m(p)$ model $m$ was used with  $P_m=p$, minus sign - sample with less than $5,000$ available, $\star$ - a reduced sample used, $t$ -  time limit reached without finding a result, $\times$ - less than $20$ observations were available}
\label{table:sumg}
\end{center}
\end{table}

\begin{table}
\begin{center}
\begin{tabular}{l|C{1.7cm}C{1.7cm}|C{1.7cm}C{1.7cm}|C{1.7cm}C{1.7cm}}
&\multicolumn{2}{c}{XOM} & \multicolumn{2}{c}{MSFT} & 
\multicolumn{2}{c}{GE} \\
& ZI & GZI & ZI & GZI& ZI & GZI \\
\hline
NASDAQ OB & $-5.49$   & $-2.60$      
		  & $-0.38$   & $-0.48$ 
		  & $n/a$   & $-4.15$  \\
NSE       & $-1.11$   & $-2.83$      
          & $-1.30$   & $-2.64$
          & $-0.12$   & $n/a$ \\
Chicago   & $-2.31$   & $-2.28$
          & $n/a$   & $-0.75$
          & $-9.72$   & $-1.47$\\
NYSE      & $n/a$   & $n/a$
          & $\times$ & $\times$
          & $n/a$   & $n/a$\\
ARCA      & $-2.18$   & $-1.75$
          & $n/a$ & $t$        
          & $t$ & $t$ \\
NASDAQ T  & $-0.84$   & $n/a$
          & $n/a$ & $\times$             
          & $n/a$   & $t$\\ 
CBOE      & $n/a$   & $-1.70$
          & $-0.97$   & $-0.34$
          & $n/a$   & $-2.12$\\
BATS      & $-1.15$   & $-4.84$
          & $n/a$   & $n/a$
          & $n/a$   & $n/a$\\
ISE       & $-1.18$   & $-2.76$
          & $n/a$   & $-6.76$
          & $-0.26$   & $n/a$
\end{tabular}

\caption{Estimated values of $\alpha_\kappa$. Notation: $n/a$ - only model $S$ came out significant, $t$ - time limit reached, $\times$ - no data available.}
\label{table:alphakappa}
\end{center}
\end{table}

\begin{table}
\begin{center}
\begin{tabular}{l|C{1.7cm}C{1.7cm}|C{1.7cm}C{1.7cm}|C{1.7cm}C{1.7cm}}
&\multicolumn{2}{c}{XOM} & \multicolumn{2}{c}{MSFT} & 
\multicolumn{2}{c}{GE} \\
& ZI & GZI & ZI & GZI& ZI & GZI \\
\hline
NASDAQ OB & $-5.37$   & $-2.60$      
		  & $-0.26$   & $-0.60$ 
		  & $n/a$   & $-4.22$  \\
NSE       & $-1.84$   & $-2.89$      
          & $-0.89$   & $-2.31$
          & $-0.61$   & $n/a$ \\
Chicago   & $-2.13$   & $-3.22$
          & $n/a$   & $-1.94$
          & $-8.28$   & $-1.77$\\
NYSE      & $n/a$   & $n/a$
          & $\times$ & $\times$
          & $n/a$   & $n/a$\\
ARCA      & $-1.77$   & $-0.77$
          & $n/a$ & $t$        
          & $t$ & $t$ \\
NASDAQ T  & $-0.67$   & $n/a$
          & $n/a$ & $\times$             
          & $n/a$   & $t$\\ 
CBOE      & $n/a$   & $-2.21$
          & $-0.03$   & $-9.24$
          & $n/a$   & $-3.39$\\
BATS      & $-0.38$   & $-4.59$
          & $n/a$   & $n/a$
          & $n/a$   & $n/a$\\
ISE       & $-0.94$   & $-2.43$
          & $n/a$   & $-3.56$
          & $-0.66$   & $n/a$
\end{tabular}

\caption{Estimated values of $\alpha_\rho$. Notation: $n/a$ - only model $S$ came out significant, $t$ - time limit reached, $\times$ - no data available. }
\label{table:alpharho}
\end{center}
\end{table}

\begin{table}
\begin{center}
\begin{tabular}{l|C{2.5cm}|C{2.5cm}|C{2.5cm}}
& MKTX & JCOM & ARL \\
\hline 
NASDAQ OB & $f^\star$     
		  & $T_1(0.06,0.06)$  
		  & $\times$  \\
NSE       & $\times$ 
          & $f^\star$ 
          & $\times$  \\
Chicago   & $\times$  
          & $\times$ 
          & $\times$ \\
NYSE      & $\times$ 
          & $\times$ 
          & $f$ \\
ARCA      & $T_2(0.40,0.45)$ 
          & $T_2(0.18,0.17)$          
          & $f^\star$   \\
NASDAQ T  & $T_2(0.44,0.46)$            
          & $T_2(0.38,0.37)$  
          & $S(0.08)^-$ \\ 
CBOE      & $S(0.92)^\star$ 
          & $T_2(<0,<0)$ 
          & $f^-$ \\
BATS      & $T_2(0.13,0.35)$            
          & $T_2(0.60,0.52)$ 
          & $f^-$ \\
ISE       & $S(0.31)^\star$ 
          & $T_2(0.39,0.49)$ 
          & $f^-$ \\
\end{tabular}

\caption{Results for ZI of the small caps. Notation:  $f$ - optimization failed, for the rest, see Table \ref{table:sumz}.}
\label{table:sumzs}
\end{center}
\end{table}

\begin{table}
\begin{center}

\begin{tabular}{l|C{1.8cm}C{1.8cm}|C{1.8cm}C{1.8cm}|C{1.8cm}C{1.8cm}}
&\multicolumn{2}{c}{MKTX} & \multicolumn{2}{c}{JCOM} & 
\multicolumn{2}{c}{ARL} \\
& GZI & ZI & GZI & ZI& GZI & ZI \\
\hline
NASDAQ OB & $\times$  & $\times$     
		  & $T_2(0.78)^-$  & $T_1(0.79)^-$
		  & $\times$  & $\times$\\
NSE       & $\times$  & $\times$ 
          & $\times$  & $\times$
          & $\times$ & $\times$\\
Chicago   & $\times$& $\times$
          & $\times$& $\times$    
          & $\times$  & $\times$\\
NYSE      & $\times$  & $\times$  
          & $\times$        & $\times$
          & $S(<0)^-$        & $T_2(<0)^-$\\
ARCA      & $T_2(0.31)^-$  & $T_2(0.25)^-$        
          & $S(0.08)$        & $T_2(0.04)$
          & $f^-$        & $f^-$\\
NASDAQ T  & $\times$  & $\times$        
          & $\times$        & $\times$    
          & $S(<0)^-$        & $T_1(0.49)^-$\\ 
CBOE      & $S(1.00)^-$   & $S(1.00)^-$      
          & $T_1(0.29)^-$ & $T_1(0.25)^-$
          & $\times$  & $\times$\\
BATS      & $S(0.41)^-$  & $S(0.39)^-$         
          & $T_2(0.65)$        & $S(0.66)$
          & $\times$    & $\times$\\
ISE       & $S(<0)^-$        & $S(<0)^-$
          & $S(<0)^-$        & $S(<0)^-$
          & $\times$  & $\times$
\end{tabular}
\caption{Results for GZI of the small caps. Notation: $f$ - optimization failed, for the rest,  see Table \ref{table:sumg}}
\label{table:sumgs}
\end{center}
\end{table}

\section{Conclusions}
\label{sec:conclusion}
A setting covering many of the existing zero intelligence (ZI) models  and its generalisation allowing for non-unit market orders and shifts of quotes (GZI) were introduced. Several variants of both the ZI and GZI settings, differing in their complexity, were tested on trade and quote data for 54 real-life stock-market pairs. It was found that, especially for liquid stocks, both the ZI and GZI models came out significant with a substantial prediction power; however, their more complicated variants did not produce significantly better predictions in comparison with a simple model with constant intensities. Finally, as the rewults of the ZI and the GZI variants are comparable, our suggestion is to use the GZI models which are more realistic and are capable of predictions of market impact.

\bibliographystyle{rQUF}
%\bibliography{/home/martin/Documents/s/smid}

\appendix

\section{Proof of Theorem \ref{th:dist}}
\label{sec:proofs}
\def\IDP{\mathcal{I}}

Before starting the proof, recall that that a Markov chain $X$ on $\R^+$ is called {\em immigration and death process} with immigration intensity $\kappa$ and death intensity $\rho$ and with initial population $A$ (we abbreviate this by $\IDP(\kappa, \rho, A)$) if its transition matrix $\Lambda=(\lambda_{i,j})$ 
has zero components except of 
$\lambda_{j,j+1}=\kappa$, $\lambda_{j,j-1}=j\rho$, $j > 0$ and if  $X_0=A$. Notice also that number of customers in an $M/M/\infty$ in queuing theory model follows an $\IDP$ process. 
It is well known (see, e.g., \cite{mandjes2011m}, Section 2) that, for any positive $t$, 
\begin{equation}\label{eq:idp}
\IDP(\kappa, \rho, A)_t \sim  \mathrm{Po}
\left(\frac{\kappa}{\rho}[1-\exp\{-\rho t\}]\right) \circ \mathrm{Bi}(A,\exp\{-\rho s\}).
\end{equation} 

Further, let us introduce the following (re)formulation of process $\Xi$ restricted to time interval $(0,t_1]$, which will be used repeatedly in the subsequent proof:
Let $u^1,v^1,u^2,v^2, \dots, u^n,v^n$ be independent uniform variables, 
independent of $\Xi_0$. By \cite{Kallenberg02}, Theorem 6.10, there exist mappings $U^1, V^1,\dots,U^n,V^n$ from $\N_0^n \times [0,1]$ to 
the space of stochastic processes on $\R^+$ such that, for each $\pi$, $U^\pi(A_0, u^\pi)|\Xi_0 \sim \IDP(\kappa_{\pi,0},\rho_{\pi,0}, A^\pi_0)$ and $V^\pi(B_0, v^\pi)|\Xi_0 \sim \IDP(\lambda_{\pi,0},\sigma_{\pi,0}, B^\pi_0)$. By \cite{Kallenberg02}, Proposition 6.13., $U^1,V^1,U^2,V^2, \dots, U^n,V^n$ are mutually conditionally independent given $\Xi_0$, so 
we may assume, without a change of distributions, 
\begin{equation}\label{eq:redef}
A_t = (U_t^1,\dots,U_t^n), \qquad B=(V_t^1,\dots,V_t^n)
\qquad \text{on $[t\leq t_1]$}.
\end{equation}
(to see it, check definitions (\ref{eq:d1}) and (\ref{eq:d2})).

Using this reformulation and noting that both $e_1$, $\Delta t_1$ are functions of $(U^\pi,V^\pi)_{\pi\leq a_0}$,
we see that 
\begin{equation}\label{eq:condi}
(t_1,e_1) \indep_{\Xi_0} (U^{\pi})_{\pi > a_i}
\end{equation} 
where $A \indep_C B$ denotes conditional indepdence of $A$ and $B$ given $C$.

Because, by the definition of the ask, 
$A^1_{[i-]}=A^2_{[i-]}=\dots=A^{a_0-1}_{[i-]}=0$ and $A^{a_0}_{[i-]}=q_0$, i.e, $A^1_{[i-]},\dots,A^{a_0}_{[i-]}$ 
are conditionally constant, we have, for any $s_{1},s_{2},\dots,s_n$,
\begin{multline}\label{eq:pu}
\P[A^{1}_{[1-]}=s_{1},A^{2}_{[1-]}=s_{2}, \dots,A^{n}_{[1-]}=s_{n}|e_1,t_1,\Xi_0]
\\
=\mathbf{1}[s_1=s_2=\dots=s_{a_{0}-1}=0, s_{a_0}=q_0] \cdot \P[U^{a_0+1}_{t_1-}=s_{a_0+1},\dots,U^n_{t_1-}=s_n|e_1,t_1,\Xi_0]
\\
=\mathbf{1}[s_1=s_2=\dots=s_{a_{0}-1}=0, s_{a_0}=q_0] 
\cdot
\prod_{\pi=a_0+1}^{n} 
\P[U^\pi_{t_1}=s_\pi]
\end{multline}
where the last ``='' follows from \cite{Hoffmann94}, (6.8.14) (with 
$R=(u^{a_0+1},u^{a_0+2},\dots,u^n)$ and $S=(e_1,t_1,\Xi_0)$) and from the fact that the jumps of $U$'s do not coincide with probability one (which guarantees that $U^\pi_{t_1}=U^\pi_{t_1-}$, $\pi > a_0$, almost sure).

As a first step of proving the assertions of the Theorem, let us deal with one-step-ahead distribution
$\P[(\Delta t_i, e_i,A_{[i-]})\in \bullet|\Xi_{[i-1]}]$:
\begin{proposition}\label{biprop} (i) For any $\tau$ and $e$,
$$
\P[\Delta t_i > \tau,e_i=e|\Xi_{[i-1]}]
= \exp\{-\gamma_{i-1} \tau\} \pi_{i}(e).
$$
(ii) For any $p$,
$$
A^p_{[i-]}|e_i, \Delta t_i, \Xi_{[i-1]} 
\sim
\begin{cases}
\dirac_{0} & \text{on $[1 \leq p< a_{[i-1]}]$}
\\
\dirac_{q_{[i-1]}} & \text{on $[p=a_{[i-1]}]$}
\\
\mathrm{Po}
\left(\phi_{p,i-1} (1-\delta_{p,i-1})
\right)
\circ
\mathrm{Bi}
\left(
A_{[i-1]}^p,\delta_{p,i-1}
\right) & \text{on $[a_{[i-1]} < p \leq n]$}. \\
\end{cases}
$$
(iii) For any $p$,
\begin{multline*}
A^p_{[i]}|e_{i+1},\Delta t_{i+1}, x_{[i]},
\Delta t_i,\Xi_{[i-1]} 
\\
\sim	
\begin{cases}
\dirac_{A_{[i]}} & \text{on $[1 \leq p \leq \max(a_{[i]},a_{[i-1]})]$}
\\
\mathrm{Po}
\left(\phi_{p,[i-1]} (1-\delta_{p,i-1})
\right)
\circ
\mathrm{Bi}
\left(
A_{[i-1]}^p,\delta_{p,i-1}
\right)
&
\text{on $[\max(a_{[i]},a_{[i-1]}) < p \leq n]$}.
\end{cases}
\end{multline*}  
\end{proposition}

\noindent
\begin{proof} Thanks to the homogeneity of the process, we may assume $i=1$.

\noindent (i)  It follows from textbook knowledge that $\Delta t_1$, being
the first jump time of Markov chain $(V^{b_0}, U^{b_0+1}, V^{a_0+1},\dots, U^{a_0})$, is exponential with rate $\gamma_0$, while $e_i$, coding a type of the chain's first jump, is (conditionally w.r.t. $\Xi_0$) independent of $\Delta t_1$ with probabilities of particular events being equal to the rates of the events' intensities to $\gamma_0$, which is formally expressed by (i).

\noindent (ii) The formula follows from (\ref{eq:pu}) and (\ref{eq:idp}).

\noindent (iii) Let $1\leq j,k \leq n$ and put $S_{j,k}=[a_{[1]}=j, a_0=k]$. Clearly,
\begin{equation}
\label{eq:r1}
x_{[1]} =\tilde x_{[1]}
\qquad \text{on $S_{j,k}$},
\end{equation}
where 
$$
\qquad \tilde x_{[1]} = (j,A^j_{[1]},\tilde b,B^{\tilde b}_{[1]}),
\qquad
\tilde b = \max\{\pi <j: B^\pi_{[1]} > 0\}.
$$
Further, it follows from the dynamics of the process 
and from the almost sure exclusivity of jumps of $\Xi$ that
\begin{multline}
\label{eq:r2}
A_{[1]} = \tilde A\quad \text{on $S_{j,k}$},
\qquad
\tilde A = (r_{j,k},A^{j\vee k+1}_{[1-]},\dots, 
A^{n}_{[1-]}),
\\
r_{j,k}=\begin{cases}
(\underbrace{0,\dots,0}_{j\times},q_{[1]}) & j \geq k \\
(\underbrace{0,\dots,0}_{j\times},q_{[1]}, 
\underbrace{0,\dots,0}_{(k-j-1)\times},q_{0}) & j < k. \\
\end{cases}
\end{multline}
Now, because $(\Delta t_1, \tilde x_{[1]})$ are functions of $U^1,V^1,\dots, U^j, V^j, \Xi_0$, we have $(\Delta t_1, \tilde x_{[1]})\in 
\sigma(U^1,V^1,\dots, U^j, V^j, \Xi_0)$ so we may use the Law of Iterated Expectation to get
\begin{multline}
\label{eq:uprobs}
\P[\tilde A \in\bullet|\tilde x_{[1]},\Delta t_1,\Xi_0]
=
\E(\P[\tilde A \in\bullet
|U^1,V^1,\dots,U^j,V^j,\Xi_0]|\tilde x_{[1]},\Delta t_1,\Xi_0)
\\
=
\E(\P[\tilde A \in\bullet
|\Xi_0]|\tilde x_{[1]},\Delta t_1,\Xi_0)
=
\P[\tilde A \in\bullet
|\Xi_0]
=p_{j,k}(\bullet;t_1,\Xi_0,r_{j,k})\quad \text{on $S_{j,k}$},
\\
p_{j,k}(s_1,\dots,s_n ;t,\Xi,r)=
\mathbf{1}[r_1=s_1,\dots,r_{j\vee k} = s_{j\vee k} ]
\prod_{\pi = j\vee k + 1}^n \P[U^\pi(\Xi)_t = s_\pi].
\end{multline}
(at the third ``='', we have used  Proposition 6.6 of \cite{Kallenberg02}, the fourth one is due to measurability of the inner term w.r.t. the condition of the expectation, the fifth one is due to (\ref{eq:pu})).

Now, by Local Property (\cite{Kallenberg02}, Lemma 6.2) applied to (\ref{eq:r1}), (\ref{eq:r2}) and (\ref{eq:uprobs}), we are getting that 
\begin{equation}
\label{eq:pau}
\P[A_{[1]} \in \bullet |x_{[1]},\Delta t_1,\Xi_0]=p_{a_{[1]},a_0}(\bullet;t_1,\Xi_0,(A^1_{[1]},A^2_{[1]}, \dots, A_{[1]}^{a_{[1]} \vee a_0}))
%\\=\mathbf{1}[s_1 = \dots = s_{a_{[1]}-1}=0,s_{a_{[1]}}=q_{[1]}]\cdot 
%\P[\IDP(\kappa_{\pi,0},\rho_{\pi,0},A^\pi_0)=s_\pi].
\end{equation}
on each $S_{j,k}$ and, as the $S_{j,k}$'s cover all the underlying probability space, the relation holds universally. 

Finally, by the strong Markov property of $\Xi$ (guaranteed, e.g., by \cite{Kallenberg02}, Theorem 12.14.), we get
$$
\P[\Delta t_2 > \tau,e_2=e|\Xi_{[1]},x_{[1]}, \Delta t_1, \Xi_{0}]
= \P[\Delta t_2 > \tau,e_2=e|\Xi_{[1]}]
$$
which is, by Proposition 6.6. of \cite{Kallenberg02}, equivalent to
$$\Delta t_2,e_2 \indep_{x_{[1]},\Delta t_1,\Xi_0} \Xi_{[1]}$$
implying, by the same Proposition with switched variables, 
$$
\P[\Xi_{[1]}\in \bullet|e_2,
\Delta t_2, x_{[1]},\Delta t_1,\Xi_0]=
\P[\Xi_{[1]}\in \bullet|
x_{[1]},\Delta t_1,\Xi_0]
$$
which, together with (\ref{eq:pau}), proves (iii). 
\end{proof}

\subsection*{Proof of the Theorem}

\subsubsection*{Ad (i)}

By the Law of Iterated Expectation, the strong Markov property of $\Xi$, and 
Proposition \ref{biprop} (i), it follows
\begin{multline}
\label{eq:markov}
\P[\Delta t_i>\tau, e_i = e|\xi_{[i-1]}]
=
\E[\P[\Delta t_i>\tau, e_i = e|\Xi_{[i-1]},\xi_{[i-1]}]|\xi_{[i-1]}]
\\
=
\E[\P[\Delta t_i>\tau, e_i = e|\Xi_{[i-1]}]|\xi_{[i-1]}]
=\E(\exp\{-\gamma_{i-1}\tau\}\pi_i(e)|\xi_{[i-1]})
=\exp\{-\gamma_{i-1}\tau\}\pi_i(e)
\end{multline}
(we could get rid of the outer conditional expectation at the last equality because its integrand is measurable with respect to its condition).

\subsubsection*{Ad (ii)}

Before deriving the distribution of $A_{[i-]}$, let us formulate an auxiliary result:

\def\F{\mathcal{F}}
\def\G{\mathcal{G}}
\begin{lemma} \label{lem:po}
Let $A,B$ be random variables, let $\F$ be as $\sigma$-algebra and let $S\in \F$. If $B|\F \sim \mathrm{Po}(\epsilon) \circ \mathrm{Bi}(\nu,\varpi)$ on $S$ and 
$A|B,\F \sim \mathrm{Po}(\zeta) \circ \mathrm{Bi}(B,\delta)$ on $S$
for some $\F$-measurable variables $\epsilon,\nu,\varpi,\zeta,\delta$ 
then
$A|\F \sim \mathrm{Po}(\zeta+\delta \epsilon) \circ \mathrm{Bi}(\nu,\delta\varpi)$ on $S$.
\end{lemma}

\begin{proof} Let $A_0,B_0,C_1, D_1,E_1,C_2, D_2,E_2 \dots$ be variables mutually conditionally independent given $\F$, such that, on $S$,  
$A_0|\F \sim \Po(\zeta)$, 
$B_0|\F\sim \mathrm{Po}(\epsilon)$, and, for each $i\in \N$,
$C_i|\F \sim \mathrm{Bernoulli}(\delta)$, 
$D_i|\F \sim \mathrm{Bernoulli}(\delta)$,
$E_i|\F \sim \mathrm{Bernoulli}(\varpi)$ (outside $S$, all the variables could be e.g., zero).
Assume, without a change of the examined distributions, that $B = B_0 + E_1 + E_2 + \dots + E_\nu$ and $A=A_0+A_1+A_2$ where
$A_1=\sum_{i=1}^{B_0} C_i$, 
$A_2=\sum_{i=1}^\nu D_i E_i$. Clearly, $A_2$ is (conditionally) Binomial with parameters $\nu$ and $\delta\varpi$. Further, by \cite{Daley03}, 2.3.4, $A_1|\F \sim \Po(\delta \epsilon)$. As $A_0, A_1, A_2$ are mutually conditionally independent given $\F$ (which is because each of them depends on independent variables) we get that $A=\Po(\zeta)\circ \Po(\delta\epsilon) \circ \Bi(\nu,\delta\varpi)$. The  Lemma now follows from the fact that a convolution of two Poisson variables is again Poisson.
\end{proof} 
\noindent
Let $1 \leq p \leq n$. First we prove, by induction, that, for any $k$,
\begin{equation}
\label{indhyp}
A^p_{[k]}|\xi_{[(k+1)-]} \sim 
\mathrm{Po}(\epsilon_{p,k}+\iota_{p,k}) \circ \mathrm{Bi}(\nu_{p,k},\varpi_{p,k}), \qquad \text{on $[a_{[k]}< p\leq n]$}.
\end{equation}

First, we show that (\ref{indhyp}) holds for $k=0$: indeed, by Proposition \ref{biprop} (i) and \cite{Kallenberg02}, Proposition 6.6., we get that  $(t_1,e_1)\indep_{x_0}\Xi_0$ (recall that $\xi_{[1-]}=x_0$), which, by the Proposition 6.6. and by (\ref{eq:zero}), gives (\ref{indhyp}) for $k=0$.
 
Now, assume (\ref{indhyp}) to hold for $k-1$. Similarly as in (\ref{eq:markov}), we get, by Proposition \ref{biprop}, (iii), that
\begin{multline}
\label{eq:paaak}
\P[
A^p_{[k]}\in \bullet|A^p_{[k-1]},\xi_{[(k+1)-]}]
\\
=
\E(\P[A^p_{[k]}\in \bullet|A^p_{[k-1]}]|
A^p_{[k-1]},\xi_{[(k+1)-]})
=
\E(\P[
A^p_{[k]}\in \bullet|\Xi_{[k-1]},e_{k},\Delta t_k]|A^p_{[k-1]},\xi_{[k-]})
\\
=
\P\left[
\mathrm{Po}
\left(\phi_{p,k-1}(1-\delta_{p,k-1})
\right)
\circ
\mathrm{Bi}
\left(
A_{[k-1]}^p,\delta_{p,k-1}
\right)\in \bullet\right]
\end{multline}
giving, by Lemma \ref{lem:po} applied to the induction hypothesis and (\ref{eq:paaak}), validity of (\ref{indhyp}) on 
$[a_{[k-1]} < p \leq n]$. 
As, on $[a_{[i-1]}=p]$, $A^p_{[k]}=q_{[k]}$ is conditionally constant, i.e.,
\begin{equation*}
\label{eq:fixed}
\P[A^p_{[k]}=\alpha|\xi_{[(k+1)-]}]
=
\mathbf{1}[\alpha = q_{[k]}] 
=\P[\Bi(A^p_{[k]},1)\circ \Po(0) = \alpha]
\end{equation*}
(\ref{indhyp}) is proved on $[a_{[i-1]}=p]$ because $k_k(p)=k$ on the set so $\epsilon_{p,k}=\iota_{p,k} = 0$, $\varpi_{p,k}=1$ and $\nu_{p,k}=A^p_{[k]}$ on the set.
As (\ref{indhyp}) on $[a_{[i-1]}>p]$ follows similarly (here, $A^p_{[k]}=0$), (\ref{indhyp}) is proved for $k$.

Having proved (\ref{indhyp}) for any $k$, we may proceed to (ii) of the Theorem: Let $1\leq p \leq n$. Similarly as above we get, using Proposition \ref{biprop} (ii),
\begin{multline*}
\label{eq:paaa}
\P[
A^p_{[i-]}\in \bullet|A^p_{[i-1]},\xi_{[i-]}]
=
\E(\P[
A^p_{[i-]}\in \bullet|\Xi_{[i-1]},\Delta t_i,e_i]|A^p_{[i-1]},\xi_{[i-]})
\\
=
\P\left[
\mathrm{Po}
\left(\phi_{p,i-1}(1-\delta_{p,i-1})
\right)
\circ
\mathrm{Bi}
\left(
A_{[i-1]}^p,\delta_{p,i-1}
\right)\in \bullet\right]
 \qquad \text{on $[a_{[i-1]} < p \leq n]$}
\end{multline*}
which implies (ii) on $[a_{[i-1]} < p \leq n]$ by Lemma \ref{lem:po}. Validity of (ii) on set $[p < a_{[i-1]}]$ on which $A^{p}_{[i-]}$ is conditionally constant, follows similarly as in the proof of (\ref{indhyp}).

\subsubsection*{Ad (iii)}
Let $i\in \N$, $1 \leq p \leq n$, denote $A^{\neg p}=(A^\pi)_{\pi \neq p}$, and assume that $A^p_{[i-1]}, A^{\neg p}_{[i-1]}$ are conditionally independent given $\xi_{[i-]}$, i.e., that
\begin{equation}
\label{eq:hypi}
\P[A^p_{[k-1]}=r,A^{\neg p}_{[k-1]}=R|
\xi_{[k-]}=\xi]
\\
=
\P[A^p_{[k-1]}=r|
\xi_{[k-]}=\xi]
\P[A^{\neg p}_{[k-1]}=R|
\xi_{[k-]}=\xi]
\end{equation}
holds for any $r,R,\xi$ and for $k = i$.

Let $I\subseteq \N_0^n$ and put 
$$
q(s,e,\tau,\Xi)
=
\P[A_{[1-]}\in I|e_1 = e,\Delta t_1 = \tau,\Xi_0 = \Xi].
$$
From (\ref{eq:pu}), we have
that
\begin{equation*}
q(I,e,\tau,\Xi_0)
=\prod_{\{1\leq \pi \leq n, I_\pi \neq \N_0 \} }
q_\pi(I_\pi,e,\tau,A^\pi_0)
\end{equation*}
for some $q_1,\dots,q_n$. So, by the Law of Iterated Expectation, strong Markov property and the homogeneity, applied gradually,
\begin{multline}
\P[A_{[i-]}\in I|\xi_{[i-]},A_{[i-1]}]=
\E(\P[A_{[i-]}\in I|\xi_{[i-]},\Xi_{[i-1]}]|\xi_{[i-]},A_{[i-1]})
\\
=
\E(\P[A_{[i-]}\in I|e_i,\Delta t_i,\Xi_{[i-1]}]|\xi_{[i-]},A_{[i-1]})
\\=
\E(q(I,e_i,\Delta t_i,\Xi_{[i-1]})|\xi_{[i-]},A_{[i-1]})
=
\prod_{\{1\leq \pi \leq n, I_\pi \neq \N_0 \} }
q_\pi(I_\pi,e_i,\Delta t_1,A^p_{[i-1]})
\label{eq:papi}
\end{multline}
(we could get rid of the expectation on the RHS because its integrand is measurable). 

Let $1\leq p \leq n$, let $J \subseteq \N_0$ and $K \subseteq \N^{n-1}_0$. By evaluating (\ref{eq:papi}) with $I=\N_0 \times \dots \times J \times \dots \times \N_0$ and with $I=K_0 \times \dots \times \N_0 \times \dots \times K_p$, we get
\begin{multline*}
\P[A^p_{[i-]}\in J,A^{\neg p}_{[i-]}\in K|\xi_{[i-]},A_{[i-1]}] 
\\=
\P[A^p_{[i-]} \in J|\xi_{[i-]}, A_{[i-1]}]
\P[A^{\neg p}_{[i-]} \in K|\xi_{[i-]}, A_{[i-1]}]
\\=
\P[A^p_{[i-]} \in J|\xi_{[i-]}, A^p_{[i-1]}]
\P[A^{\neg p}_{[i-]} \in K|\xi_{[i-]}, A^{\neg p}_{[i-1]}]
\end{multline*}
(we could write the last ``='' because the right-hand side of (\ref{eq:papi}) does not depend on no $A^\pi$ with $I_\pi=\N_0$).
Thus, by the Complete Probability Theorem,
\begin{multline}
\label{eq:complete}
\P[A^p_{[i-]}\in J,A^{\neg p}_{[i-]}\in K| \xi_{[i-]}=\xi] 
\\= 
\sum_r \sum_R 
\P[A^p_{[i-]}\in J,A^{\neg p}_{[i-]}\in K|A^p_{[i-1]}=r,A^{\neg p}_{[i-1]}=R,
\xi_{[i-]}=\xi
]\\
\qquad \qquad \qquad \qquad \qquad 
\qquad \qquad \qquad \qquad \times
\P[A^p_{[i-1]}=r,A^{\neg p}_{[i-1]}=R|\xi_{[i-]}=\xi]
\\=
\left(\sum_r 
\P[A^p_{[i-]}\in J|A^p_{[i-1]}=r,\xi_{[i-]}=\xi]
\P[A^p_{[i-1]}=r|\xi_{[i-]}=\xi]\right)
\\
\qquad \times \left(\sum_R 
\P[A^{\neg p}_{[i-]}\in K|A^{\neg p}_{[i-1]}=R,\xi_{[i-]}=\xi]
\P[A^{\neg p}_{[i-1]}=R|\xi_{[i-]}=\xi]\right)
\\
=
\P[A^p_{[i-]} \in J|\xi_{[i-]}=\xi]
\P[A^{\neg p}_{[i-]} \in K|\xi_{[i-]}=\xi]
\end{multline}
proving conditional independence of $A^p_{[i-]}$ of $A^{\neg p}_{[i-]}$ given $\xi_{[i-]}$ for any $p$, which implies mutual conditional independence of 
$A^1_{[i-]},A^2_{[i-]},\dots,A^n_{[i-]}$ given $\xi_{[i-]}$ (see \cite{Kallenberg02}, p. 109 below).

It remains to prove (\ref{eq:hypi}) for $k=i+1$. To this end, put $S_{j,k,r,s}=[a_{[i]}=j,a_{[i-1]}=k,q_{[i]}=r,q_{[i-1]}=s]$ for any $j,k$, and observe that, on each $S_{j,k,r,s}$,  $A^1_{[i]},\dots,A^{j \vee k}_{[i]}$ are conditionally constant while $A^\pi_{[i]}=A^\pi_{[i-]}$, for $\pi > j \vee k$. As $A^{j\vee k}_{[i-]},\dots,A^n_{[i-]}$ are mutually conditionally independent given $\xi_{[i]}$ by (\ref{eq:complete}) and as (conditional) constants are trivially (conditionally) independent of any random element, we get that $A^{1}_{[i]},\dots,A^n_{[i]}$ are mutually conditionally independent given $\xi_{[i]}$ on $S_{j,k,r,s}$ by Local Property (\cite{Kallenberg02}, Lemma 6.2), and, as $S$'s cover all the underlying space, the conditional independence holds universally. Relation (\ref{eq:hypi}) now follows from (i) of Proposition \ref{biprop} (the same way as in proof of (iii) of the Proposition). 

Finally, as (\ref{eq:hypi}) holds for $k=1$ by (\ref{eq:zero}),  we have just proved (iii) by induction.

\section{MLE Estimation}
\label{sec:mle}

In the present Section, asymptotic properties of the MLE estimator are proved both for the ZI and the GZI model. 

Before starting, let us formulate an auxiliary result:

\begin{lemma}\label{erglemma}
If $X$ is a continuous time stationary ergodic Markov chain in countable space $S$ then $Y=(\Delta \tau_i,X_{\tau_i})_{i\in \N}$, where $\tau_1,\tau_2,\dots$ are the jump times of $X$, is a stationary ergodic stochastic process. 
\end{lemma} 
\def\indep#1{{\perp\hspace{-2mm}\perp}#1}
\begin{proof} Denote $\Lambda(x)$ the intensity of the first jump of $X$ given that $X_0=x$. From the strong Markov property (Theorem 12.14 of \cite{Kallenberg02}), from Lemma 12.16 of \cite{Kallenberg02} and from the scalability of exponential distribution we have that $U_1,U_2,\dots$, where $U_i = \Delta \tau_i \quad / \quad \Lambda(X_{\tau_{i-1}})$, $i \in \N$, is a sequence of i.i.d. (unit exponentially) distributed variables, independent of $X_0,X_{\tau_1},\dots$. 
As $U_n$ - being an i.i.d. sequence - is a strong mixing and $X_{\tau_n}$ is a strong mixing by \cite{Bradley05}, process
$Z_n:=(X_{\tau_n},U_n)$ is a strong mixing (note that, for any $A=A^1\times A^2$, $B=B^1\times B^2$, $A^1,B^1\in S^\Z$, $A^2,B^2\in \R^\Z$, it holds that 
$\P(Z \in A\cap T_n B)=\P(Z^1 \in (A^1 \cap T_n B^1))\P(Z^2 \in (A^2 \cap T_n B^2))\rightarrow \P(Z^1 \in A^1)\P(Z^1\in T_n B^1)\P(Z^2 \in A^2)\P(Z^2\in T_n B^2)=\P(Z \in A)\P(Z\in B)$ where $T_n$ is a shift operator; the case of general $A,B$ follows from their approximation by rectangles) we get the ergodicity of $Z$ by the well known fact that strong mixing implies ergodicity. Finally, as $Y_n$ is a function of $Z_n$, the Lemma is proved.
\end{proof}
\noindent It follows from the Lemma that, once $\Xi_{t}$ is understood as a process with $t\in(-\infty, \infty)$ then, thanks to its ergodicity (see Proposition \ref{ergprop}) process $\xi_{[i-]}$, $i \in \Z$, is also ergodic. Hence by
Birkhoff-Khinchin theorem (\cite{cornfeld1982}, Appendix 3), for any measurable function $f$,
\begin{equation}
\label{eq:erg}
\frac{1}{n}\sum_{i=1}^n f(\xi_{[i-]}) \rightarrow \E f(\xi_{[0-]}) 
\end{equation}
in probability given that the expression written on the RHS exists.

\def\pdim{k}
Now, let $\theta_0 \in \R^\pdim$ be a vector of true parameters. Denote
$$
B_N = \E\left[L_N(\theta_0)
\right],
\qquad
L_N(\theta)=\sum_{i=1}^N
\left(
\frac{g_i^{(m)}(\theta)
g^{(n)}_i(\theta)}{g^2_i(\theta)}
\right)_{m=1,\dots,\pdim,n=1,\dots,\pdim},
$$
($B_N$ being called Fischer information matrix) where $g^{(m)}$ stands for $\frac{\partial}{\partial \theta^m}g$. Substituting mappings assigning $\xi_{[i-]}\rightarrow g_i$ and $\xi_{[i-]}\rightarrow \nabla g_i$ for $f$ in (\ref{eq:erg}), we get, by (\ref{eq:erg}), that 
\begin{equation}
\label{eq:lne}
\frac{1}N L_N =M(\theta_0)
\end{equation}
for some non-random matrix $M$ which, being a limit of positive semidefinite matrices, is also positive semidefinite. By taking expectations on both sides of (\ref{eq:lne}), we further get
\begin{equation}
\label{eq:bne}
\frac{1}N B_N =M(\theta_0).
\end{equation}
Now put
$$
I_N(\theta) 
= 
\sum_{i=1}^N \nabla^2_{\theta} \log(g_i(\theta))
$$
and observe that 
$$
I_N(\theta) 
= 
K_N(\theta)
-
L_N(\theta),\qquad 
K_N = \sum_{i=1}^N \frac{1}{g_i(\theta)}
\nabla^2_\theta
g_i(\theta)
$$
As $\int g_i(\theta) = 1$ by the definition of density, any of the integral's first or second derivatives has to be zero, and, in particular, $\nabla^2 \int g(\theta)= \int \nabla^2g(\theta)=0$ (we may interchange the integral and derivative in our discrete case) so
$$
\E K_N(\theta_0) = \int \frac{1}{g_i(\theta_0)}
\nabla^2 g_i(\theta_0)
 g_i(\theta_0) = 
\int \nabla^2 g_i(\theta_0) =  0
$$
implying
\begin{equation}
\label{eq:inc}
-\frac{1}{N} I_N(\theta_0) \rightarrow M(\theta_0).
\end{equation}
If we now, as usual, assume $M(\theta_0)$ to be regular, then, by (\ref{eq:bne}) and (\ref{eq:inc}), 
\begin{equation}
-N^{-1/2} B_N^{-1/2} I_N(\theta_0)=
[(N^{-1}B_N)^{-1}(-N^{-1}I_N(\theta_0))]^{1/2}[N^{-1}(-I_N(\theta_0))]^{1/2}\rightarrow 
M(\theta_0)^{-1/2}.
\label{eq:bnest}
\end{equation}
Let $\|\theta-\theta_0\|=\delta$ for some $\delta$ and, for any matrix $J$, denote 
$F_N(J) = -(\theta-\theta_0)^T[N^{-1/2} B_N^{-1/2}J](\theta-\theta_0)$. By (\ref{eq:bnest}) and basic linear algebra, 
\begin{equation}
\label{eq:bnlim}
\lim_N F_N(I_N(\theta_0))
=(\theta-\theta_0)^T M(\theta_0)^{-1/2}(\theta-\theta_0)
\geq 
\delta^2 \lambda_{\min}
\end{equation}
where $\lambda_{\min}$ is the smallest eigenvalue of $M(\theta_0)^{1/2}$.
If, in addition, the parameter space is open and both $\kappa$ and $\rho$ are twice continuously differentiable with respect to all the parameters (which is true for all the versions of the model we use), then $\frac{1}{N} I_N(\theta)$ is asymptotically Lipschitz on 
$\{\theta:\|\theta - \theta_0\|\leq \delta\}$ so
\begin{multline}
\lim_N \left| F_N(I_N(\theta)-I_N(\theta_0)) \right|
\\
\leq 
\lim_N \| \theta-\theta_0 \| 
\| N^{-1/2}   
B_N^{-1 /2} \| \|I_N(\theta)-I_N(\theta_0)\|
\|\theta-\theta_0\|
 \leq K_\delta \| M(\theta_0)^{-1/2}\|\delta^3
\label{eq:lip}
\end{multline}
where $\|\bullet \|$ is a suitable norm and $K_\delta$ is a Lipschitz constant, further implying,
by using the fact that $|a+b|>|a|-|b|$ together with (\ref{eq:bnlim}) and (\ref{eq:bnlim}), we get
$$
\lim_N F_N(I_N(\theta))\geq \lim_N F_N(I_N(\theta_0)) - \lim_N | F_N(I_N(\theta)-I_N(\theta_0))|
\geq 
\delta^2 g(\delta).
$$
where $g(\delta) = (\lambda_{\min} - K_\delta \| M(\theta_0)^{-1/2}\| \delta)$.
From the differentiability it follows that $K_\delta \rightarrow K$ as $\delta\rightarrow 0$ for some $K$ so there has to exist $\Delta$ and $g_0>0$ such that $g(\delta) \geq g_0$ for all $\delta \leq \Delta$. Weak consistency now follows from (2.3) of \cite{crowder1976} with $c_n=n^\beta$ for suitable $\beta<1$.

For the asymptotic normality it suffices, by \cite{crowder1976}, (4.13) and the considerations explained below, that the absolute $k$-th moments of the observations are bounded for a certain $k>2$. This, however, can be easily achieved by determining a large enough constant $L$ and excluding from the sample any observation $(x_{[i]},\tau_{[i]})$ with $\E(\|(x_{[i]},\tau_{[i]})\|^k|\xi_{[i-]})>L$.

\section{Detailed Results}
\label{sec:detail}

\subsection*{Notation}
$N$ - sample size,  $\overline \Delta a^+$ - average jump of $a$ up, $\bar q$ - average market order volume. The number of stars stands for signification on levels $0.05$, $0.01$ and $0.001$, respectively. 

\subsection{ZI Model}
\label{zires}

\tiny
\begin{tabular}{c|c|c} \hline 
\begin{minipage}[t]{4.5cm}
\vspace{3mm}
{ \normalsize XOM on NASDAQ OB}
\\

\smallskip
$N=5000$, 
$\overline{a}^+=1.308000$, $\overline q=1.000000$\\

\smallskip
$P_S=-0.150$ $P_m=-0.147$
\\

 \smallskip
$\kappa_0=568.7488(140.9081)^{***}$\\$\alpha_\kappa=-5.4933(0.3325)^{***}$\\$\rho_0=510.4026(126.1866)^{***}$\\$\alpha_\rho=-5.3680(0.3256)^{***}$\\\smallskip
\end{minipage}
&\begin{minipage}[t]{4.5cm}
\vspace{3mm}
{ \normalsize XOM on NSE}
\\

\smallskip
$N=5000$, 
$\overline{a}^+=48.6565$\\

\smallskip
$P_S=0.897$ $P_m=0.910$
\\

 \smallskip
$\kappa_0=0.5002(0.0227)^{***}$\\$\kappa_1=0.0029(0.0007)^{***}$\\$\alpha_\kappa=-1.1057(0.0741)^{***}$\\$\rho_0=6.8241(0.1139)^{***}$\\$\rho_1=0.4923(0.0244)^{***}$\\$\alpha_\rho=-1.8369(0.0464)^{***}$\\\smallskip
\end{minipage}
&\begin{minipage}[t]{4.5cm}
\vspace{3mm}
{ \normalsize XOM on Chicago}
\\

\smallskip
$N=119$, 
$\overline{a}^+=6.9466$\\

\smallskip
$P_m=0.466$
\\

 \smallskip
$\kappa_0=0.0187(0.0018)^{***}$\\$\kappa_1=0.0800(0.0093)^{***}$\\$\alpha_\kappa=-2.3140(0.0786)^{***}$\\$\rho_0=0.0128(0.0013)^{***}$\\$\rho_1=0.1151(0.0131)^{***}$\\$\alpha_\rho=-2.1322(0.0710)^{***}$\\\smallskip
\end{minipage}
\end{tabular}\\
\begin{tabular}{c|c|c} \hline 
\begin{minipage}[t]{4.5cm}
\vspace{3mm}
{ \normalsize XOM on NYSE}
\\

\smallskip
$N=5000$, 
$\overline{a}^+=1.6307$\\

\smallskip
$P_S=0.332$ \\

 \smallskip
$\kappa_0=4.6981(0.0163)^{***}$\\$\rho_0=0.9183(0.0029)^{***}$\\\smallskip
\end{minipage}
&\begin{minipage}[t]{4.5cm}
\vspace{3mm}
{ \normalsize XOM on ARCA}
\\

\smallskip
$N=5000$, 
$\overline{a}^+=1.1156$\\

\smallskip
$P_S=0.915$ $P_m=0.917$
\\

 \smallskip
$\kappa_0=50.8499(0.4649)^{***}$\\$\alpha_\kappa=-2.1805(0.0213)^{***}$\\$\rho_0=8.1381(0.0680)^{***}$\\$\alpha_\rho=-1.7708(0.0192)^{***}$\\\smallskip
\end{minipage}
&\begin{minipage}[t]{4.5cm}
\vspace{3mm}
{ \normalsize XOM on NASDAQ T}
\\

\smallskip
$N=5000$, 
$\overline{a}^+=1.0544$\\

\smallskip
$P_S=0.532$ $P_m=0.532$
\\

 \smallskip
$\kappa_0=339.0192(6.6030)^{***}$\\$\alpha_\kappa=-0.8425(0.1939)^{***}$\\$\rho_0=52.1695(0.9973)^{***}$\\$\alpha_\rho=-0.6747(0.1931)^{***}$\\\smallskip
\end{minipage}
\end{tabular}\\
\begin{tabular}{c|c|c} \hline 
\begin{minipage}[t]{4.5cm}
\vspace{3mm}
{ \normalsize XOM on CBOE}
\\

\smallskip
$N=5000$, 
$\overline{a}^+=64.8869$\\

\smallskip
$P_S=-0.019$ \\

 \smallskip
$\kappa_0=0.0000(0.0000)^{***}$\\$\rho_0=0.0000(0.0000)^{***}$\\\smallskip
\end{minipage}
&\begin{minipage}[t]{4.5cm}
\vspace{3mm}
{ \normalsize XOM on BATS}
\\

\smallskip
$N=5000$, 
$\overline{a}^+=1.0255$\\

\smallskip
$P_S=0.827$ $P_m=0.859$
\\

 \smallskip
$\kappa_0=649.2523(12.7981)^{***}$\\$\kappa_1=7.6678(0.7012)^{***}$\\$\alpha_\kappa=-1.1529(0.1952)^{***}$\\$\rho_0=97.1699(1.8800)^{***}$\\$\rho_1=1.2875(0.1017)^{***}$\\$\alpha_\rho=-0.3755(0.2070)^{*}$\\\smallskip
\end{minipage}
&\begin{minipage}[t]{4.5cm}
\vspace{3mm}
{ \normalsize XOM on ISE}
\\

\smallskip
$N=5000$, 
$\overline{a}^+=8.1611$\\

\smallskip
$P_S=0.674$ $P_m=0.656$
\\

 \smallskip
$\kappa_0=3.6649(0.0790)^{***}$\\$\kappa_1=0.1846(0.0062)^{***}$\\$\alpha_\kappa=-1.1812(0.0154)^{***}$\\$\rho_0=8.7439(0.1230)^{***}$\\$\rho_1=0.6432(0.0117)^{***}$\\$\alpha_\rho=-0.9377(0.0122)^{***}$\\\smallskip
\end{minipage}
\end{tabular}\\
\begin{tabular}{c|c|c} \hline 
\begin{minipage}[t]{4.5cm}
\vspace{3mm}
{ \normalsize MSFT on NASDAQ OB}
\\

\smallskip
$N=5000$, 
$\overline{a}^+=1.0820$\\

\smallskip
$P_S=0.703$ $P_m=0.700$
\\

 \smallskip
$\kappa_0=4.7416(0.0174)^{***}$\\$\alpha_\kappa=-0.3883(0.0325)^{***}$\\$\rho_0=0.1427(0.0005)^{***}$\\$\alpha_\rho=-0.2643(0.0319)^{***}$\\\smallskip
\end{minipage}
&\begin{minipage}[t]{4.5cm}
\vspace{3mm}
{ \normalsize MSFT on NSE}
\\

\smallskip
$N=5000$, 
$\overline{a}^+=1.3504$\\

\smallskip
$P_S=0.866$ $P_m=0.889$
\\

 \smallskip
$\kappa_0=0.0088(0.0004)^{***}$\\$\kappa_1=0.0129(0.0020)^{***}$\\$\alpha_\kappa=-1.3047(0.0966)^{***}$\\$\rho_0=0.0094(0.0002)^{***}$\\$\rho_1=0.0050(0.0004)^{***}$\\$\alpha_\rho=-0.8890(0.0653)^{***}$\\\smallskip
\end{minipage}
&\begin{minipage}[t]{4.5cm}
\vspace{3mm}
{ \normalsize MSFT on Chicago}
\\

\smallskip
$N=5000$, 
$\overline{a}^+=3.6887$\\

\smallskip
$P_S=0.929$ \\

 \smallskip
$\kappa_0=0.0020(0.0000)^{***}$\\$\rho_0=0.0007(0.0000)^{***}$\\\smallskip
\end{minipage}
\end{tabular}\\
\begin{tabular}{c|c|c} \hline 
\begin{minipage}[t]{4.5cm}
\vspace{3mm}
{ \normalsize MSFT on NYSE}
\\

\smallskip
$N=0$, 
$\overline{a}^+=0.0000$\\

\smallskip
No data
\end{minipage}
&\begin{minipage}[t]{4.5cm}
\vspace{3mm}
{ \normalsize MSFT on ARCA}
\\

\smallskip
$N=1000^\star$, 
$\overline{a}^+=1.000000$\\

\smallskip
$P_S=-inf$  
\\

 \smallskip
$\kappa_0=1.2963(0.0034)^{***}$\\$\rho_0=0.0105(0.0000)^{***}$\\\smallskip
\end{minipage}

&
\begin{minipage}[t]{4.5cm}
\vspace{3mm}
{ \normalsize MSFT on NASDAQ T}
\\

\smallskip
$N=1000^\star$, 
$\overline{a}^+=1.000909$\\

\smallskip
$P_S=1.000$ 
\\

 \smallskip
$\kappa_0=1.1965(0.0028)^{***}$\\$\rho_0=0.0064(0.0000)^{***}$\\\smallskip
\end{minipage}
\end{tabular}\\
\begin{tabular}{c|c|c} \hline 
\begin{minipage}[t]{4.5cm}
\vspace{3mm}
{ \normalsize MSFT on CBOE}
\\

\smallskip
$N=5000$, 
$\overline{a}^+=19.3225$\\

\smallskip
$P_S=0.359$ $P_m=0.346$
\\

 \smallskip
$\kappa_0=0.2417(0.0026)^{***}$\\$\kappa_1=0.0606(0.0006)^{***}$\\$\alpha_\kappa=-0.9796(0.0050)^{***}$\\$\rho_0=0.0825(0.0009)^{***}$\\$\rho_1=0.0168(0.0002)^{***}$\\$\alpha_\rho=-0.0336(0.0055)^{***}$\\\smallskip
\end{minipage}
&\begin{minipage}[t]{4.5cm}
\vspace{3mm}
{ \normalsize MSFT on BATS}
\\

\smallskip
$N=5000$, 
$\overline{a}^+=1.0002$\\

\smallskip
$P_S=1.000$ 
\\

 \smallskip
$\kappa_0=1.3741(0.0010)^{***}$\\$\rho_0=0.0245(0.0000)^{***}$\\\smallskip
\end{minipage}
&\begin{minipage}[t]{4.5cm}
\vspace{3mm}
{ \normalsize MSFT on ISE}
\\

\smallskip
$N=5000$, 
$\overline{a}^+=1.3325$\\

\smallskip
$P_S=0.656$ \\

 \smallskip
$\kappa_0=1.9125(0.0027)^{***}$\\$\rho_0=0.0551(0.0001)^{***}$\\\smallskip
\end{minipage}
\end{tabular}\\
\begin{tabular}{c|c|c} \hline 
\begin{minipage}[t]{4.5cm}
\vspace{3mm}
{ \normalsize GE on NASDAQ OB}
\\

\smallskip
$N=5000$, 
$\overline{a}^+=1.0411$\\

\smallskip
$P_S=0.481$ \\

 \smallskip
$\kappa_0=1.3310(0.0064)^{***}$\\$\rho_0=0.0872(0.0003)^{***}$\\\smallskip
\end{minipage}
&\begin{minipage}[t]{4.5cm}
\vspace{3mm}
{ \normalsize GE on NSE}
\\

\smallskip
$N=5000$, 
$\overline{a}^+=1.0415$\\

\smallskip
$P_S=0.763$ $P_m=0.757$
\\

 \smallskip
$\kappa_0=0.5888(0.0024)^{***}$\\$\alpha_\kappa=-0.1178(0.0061)^{***}$\\$\rho_0=0.0425(0.0001)^{***}$\\$\alpha_\rho=-0.6079(0.0051)^{***}$\\\smallskip
\end{minipage}
&\begin{minipage}[t]{4.5cm}
\vspace{3mm}
{ \normalsize GE on Chicago}
\\

\smallskip
$N=5000$, 
$\overline{a}^+=2.8520$\\

\smallskip
$P_S=0.714$ $P_m=0.758$
\\

 \smallskip
$\kappa_0=77.0015(0.1864)^{***}$\\$\alpha_\kappa=-9.7247(0.0177)^{***}$\\$\rho_0=1.4404(0.0035)^{***}$\\$\alpha_\rho=-8.2762(0.0164)^{***}$\\\smallskip
\end{minipage}
\end{tabular}\\
\begin{tabular}{c|c|c} \hline 
\begin{minipage}[t]{4.5cm}
\vspace{3mm}
{ \normalsize GE on NYSE}
\\

\smallskip
$N=5000$, 
$\overline{a}^+=1.0289$\\

\smallskip
$P_S=0.534$
\\

 \smallskip
$\kappa_0=1.4353(0.0019)^{***}$\\$\rho_0=0.0432(0.0000)^{***}$\\\smallskip
\end{minipage}
&\begin{minipage}[t]{4.5cm}
\vspace{3mm}
{ \normalsize GE on ARCA}
\\

\smallskip
$N=5000$, 
$\overline{a}^+=1.0031$\\

\smallskip
Timeout reached.
\end{minipage}
&\begin{minipage}[t]{4.5cm}
\vspace{3mm}
{ \normalsize GE on NASDAQ T}
\\

\smallskip
$N=1000^\star$, 
$\overline{a}^+=1.000000$\\

\smallskip
$P_S=-inf$ 
\\

 \smallskip
$\kappa_0=1.4767(0.0050)^{***}$\\$\rho_0=0.0475(0.0000)^{***}$\\\smallskip
\end{minipage}
\end{tabular}\\
\begin{tabular}{c|c|c} \hline 
\begin{minipage}[t]{4.5cm}
\vspace{3mm}
{ \normalsize GE on CBOE}
\\

\smallskip
$N=5000$, 
$\overline{a}^+=5.2058$\\

\smallskip
$P_S=0.672$ \\

 \smallskip
$\kappa_0=0.2249(0.0007)^{***}$\\$\rho_0=0.1067(0.0005)^{***}$\\\smallskip
\end{minipage}
&
\begin{minipage}[t]{4.5cm}
\vspace{3mm}
{ \normalsize GE on BATS}
\\

\smallskip
$N=1000^\star$, 
$\overline{a}^+=1.000000$\\

\smallskip
$P_S=1.00$
\\

 \smallskip
$\kappa_0=1.3314(0.0028)^{***}$\\$\rho_0=0.0343(0.0000)^{***}$\\\smallskip
\end{minipage}

&\begin{minipage}[t]{4.5cm}
\vspace{3mm}
{ \normalsize GE on ISE}
\\

\smallskip
$N=1000^\star$, 
$\overline{a}^+=3.272727$\\

\smallskip
$P_S=0.930$ $P_m=0.943$
\\

 \smallskip
$\kappa_0=32.9135(0.0937)^{***}$\\$\alpha_\kappa=-0.2634(0.0164)^{***}$\\$\rho_0=0.2761(0.0007)^{***}$\\$\alpha_\rho=-0.6564(0.0156)^{***}$\\\smallskip
\end{minipage}
\end{tabular}\\

\begin{tabular}{c|c|c} \hline 
\begin{minipage}[t]{4.5cm}
\vspace{3mm}
{ \normalsize MKTX on NASDAQ OB}
\\

\smallskip
$N=1303$, 
$\overline{a}^+=16.015342$\\

\smallskip
Optimization failed
\end{minipage}
&\begin{minipage}[t]{4.5cm}
\vspace{3mm}
{ \normalsize MKTX on NSE}
\\

\smallskip
$N=16$, 
$\overline{a}^+=7.4444$\\

\smallskip
No data
\end{minipage}
&\begin{minipage}[t]{4.5cm}
\vspace{3mm}
{ \normalsize MKTX on Chicago}
\\

\smallskip
$N=0$, 
$\overline{a}^+=0.0000$\\

\smallskip
No data
\end{minipage}
\end{tabular}\\
\begin{tabular}{c|c|c} \hline 
\begin{minipage}[t]{4.5cm}
\vspace{3mm}
{ \normalsize MKTX on NYSE}
\\

\smallskip
$N=0$, 
$\overline{a}^+=0.0000$\\

\smallskip
No data
\end{minipage}
&\begin{minipage}[t]{4.5cm}
\vspace{3mm}
{ \normalsize MKTX on ARCA}
\\

\smallskip
$N=5000$, 
$\overline{a}^+=5.0324$\\

\smallskip
$P_S=0.450$ $P_m=0.403$
\\

 \smallskip
$\kappa_0=0.2160(0.0018)^{***}$\\$\kappa_1=0.0393(0.0007)^{***}$\\$\alpha_\kappa=-0.5419(0.0075)^{***}$\\$\rho_0=0.1022(0.0006)^{***}$\\$\rho_1=0.1463(0.0019)^{***}$\\$\alpha_\rho=-1.1451(0.0065)^{***}$\\\smallskip
\end{minipage}
&\begin{minipage}[t]{4.5cm}
\vspace{3mm}
{ \normalsize MKTX on NASDAQ T}
\\

\smallskip
$N=5000$, 
$\overline{a}^+=3.3373$\\

\smallskip
$P_S=0.457$ $P_m=0.442$
\\

 \smallskip
$\kappa_0=0.0307(0.0008)^{***}$\\$\kappa_1=0.0340(0.0012)^{***}$\\$\alpha_\kappa=-0.8618(0.0284)^{***}$\\$\rho_0=0.0475(0.0006)^{***}$\\$\rho_1=0.0587(0.0018)^{***}$\\$\alpha_\rho=-0.7235(0.0261)^{***}$\\\smallskip
\end{minipage}
\end{tabular}\\
\begin{tabular}{c|c|c} \hline 
\begin{minipage}[t]{4.5cm}
\vspace{3mm}
{ \normalsize MKTX on CBOE}
\\

\smallskip
$N=1000$, 
$\overline{a}^+=208.296364$\\

\smallskip
$P_S=0.928$ \\

 \smallskip
$\kappa_0=0.0000(0.0000)^{***}$\\$\rho_0=0.0000(0.0000)^{***}$\\\smallskip
\end{minipage}
&\begin{minipage}[t]{4.5cm}
\vspace{3mm}
{ \normalsize MKTX on BATS}
\\

\smallskip
$N=5000$, 
$\overline{a}^+=14.3535$\\

\smallskip
$P_S=0.352$ $P_m=0.134$
\\

 \smallskip
$\kappa_0=1.4282(0.0185)^{***}$\\$\kappa_1=0.0056(0.0001)^{***}$\\$\alpha_\kappa=-0.0300(0.0072)^{***}$\\$\rho_0=1.7276(0.0180)^{***}$\\$\rho_1=0.0311(0.0006)^{***}$\\$\alpha_\rho=-0.9594(0.0081)^{***}$\\\smallskip
\end{minipage}
&\begin{minipage}[t]{4.5cm}
\vspace{3mm}
{ \normalsize MKTX on ISE}
\\

\smallskip
$N=5000$, 
$\overline{a}^+=12.6251$\\

\smallskip
$P_S=0.308$ \\

 \smallskip
$\kappa_0=0.0050(0.0000)^{***}$\\$\rho_0=0.0048(0.0000)^{***}$\\\smallskip
\end{minipage}
\end{tabular}\\
\begin{tabular}{c|c|c} \hline 
\begin{minipage}[t]{4.5cm}
\vspace{3mm}
{ \normalsize JCOM on NASDAQ OB}
\\

\smallskip
$N=5000$, 
$\overline{a}^+=2.7002$\\

\smallskip
$P_S=0.055$ $P_m=0.056$
\\

 \smallskip
$\kappa_0=0.5804(0.0379)^{***}$\\$\alpha_\kappa=-0.4535(0.0879)^{***}$\\$\rho_0=1.5632(0.0984)^{***}$\\$\alpha_\rho=-0.4486(0.0870)^{***}$\\\smallskip
\end{minipage}
&\begin{minipage}[t]{4.5cm}
\vspace{3mm}
{ \normalsize JCOM on NSE}
\\

\smallskip
$N=40$, 
$\overline{a}^+=9.1136$\\

\smallskip
Optimization failed
\end{minipage}
&\begin{minipage}[t]{4.5cm}
\vspace{3mm}
{ \normalsize JCOM on Chicago}
\\

\smallskip
$N=0$, 
$\overline{a}^+=505.0000$\\

\smallskip
No data
\end{minipage}
\end{tabular}\\
\begin{tabular}{c|c|c} \hline 
\begin{minipage}[t]{4.5cm}
\vspace{3mm}
{ \normalsize JCOM on NYSE}
\\

\smallskip
$N=0$, 
$\overline{a}^+=0.0000$\\

\smallskip
No data
\end{minipage}
&\begin{minipage}[t]{4.5cm}
\vspace{3mm}
{ \normalsize JCOM on ARCA}
\\

\smallskip
$N=5000$, 
$\overline{a}^+=1.7622$\\

\smallskip
$P_S=0.169$ $P_m=0.177$
\\

 \smallskip
$\kappa_0=9.0255(0.3157)^{***}$\\$\kappa_1=15.3048(1.4759)^{***}$\\$\alpha_\kappa=-2.6170(0.0562)^{***}$\\$\rho_0=7.4636(0.2667)^{***}$\\$\rho_1=12.9801(1.2344)^{***}$\\$\alpha_\rho=-2.4110(0.0565)^{***}$\\\smallskip
\end{minipage}
&\begin{minipage}[t]{4.5cm}
\vspace{3mm}
{ \normalsize JCOM on NASDAQ T}
\\

\smallskip
$N=5000$, 
$\overline{a}^+=2.0356$\\

\smallskip
$P_S=0.376$ $P_m=0.382$
\\

 \smallskip
$\kappa_0=0.4405(0.0081)^{***}$\\$\kappa_1=0.3113(0.0067)^{***}$\\$\alpha_\kappa=-1.0299(0.0226)^{***}$\\$\rho_0=0.3510(0.0053)^{***}$\\$\rho_1=0.2242(0.0042)^{***}$\\$\alpha_\rho=-1.3310(0.0243)^{***}$\\\smallskip
\end{minipage}
\end{tabular}\\
\begin{tabular}{c|c|c} \hline 
\begin{minipage}[t]{4.5cm}
\vspace{3mm}
{ \normalsize JCOM on CBOE}
\\

\smallskip
$N=5000$, 
$\overline{a}^+=113.4504$\\

\smallskip
$P_S=-0.196$ $P_m=-0.221$
\\

 \smallskip
$\kappa_0=0.0007(0.0001)^{***}$\\$\kappa_1=0.0876(0.0012)^{***}$\\$\alpha_\kappa=-1.6064(0.0057)^{***}$\\$\rho_0=0.0993(0.0170)^{***}$\\$\rho_1=0.1877(0.0025)^{***}$\\$\alpha_\rho=-0.6117(0.0052)^{***}$\\\smallskip
\end{minipage}
&\begin{minipage}[t]{4.5cm}
\vspace{3mm}
{ \normalsize JCOM on BATS}
\\

\smallskip
$N=5000$, 
$\overline{a}^+=2.6493$\\

\smallskip
$P_S=0.526$ $P_m=0.605$
\\

 \smallskip
$\kappa_0=50.9194(1.6604)^{***}$\\$\kappa_1=25.2161(1.6422)^{***}$\\$\alpha_\kappa=-3.1496(0.0218)^{***}$\\$\rho_0=47.6554(1.5239)^{***}$\\$\rho_1=22.8284(1.4636)^{***}$\\$\alpha_\rho=-2.6040(0.0205)^{***}$\\\smallskip
\end{minipage}
&\begin{minipage}[t]{4.5cm}
\vspace{3mm}
{ \normalsize JCOM on ISE}
\\

\smallskip
$N=5000$, 
$\overline{a}^+=23.2402$\\

\smallskip
$P_S=0.487$ $P_m=0.390$
\\

 \smallskip
$\kappa_0=0.0167(0.0003)^{***}$\\$\kappa_1=0.0095(0.0002)^{***}$\\$\alpha_\kappa=-0.4696(0.0063)^{***}$\\$\rho_0=0.0068(0.0001)^{***}$\\$\rho_1=0.1036(0.0010)^{***}$\\$\alpha_\rho=-0.7617(0.0034)^{***}$\\\smallskip
\end{minipage}
\end{tabular}\\
\begin{tabular}{c|c|c} \hline 
\begin{minipage}[t]{4.5cm}
\vspace{3mm}
{ \normalsize ARL on NASDAQ OB}
\\

\smallskip
$N=0$, 
$\overline{a}^+=0.0000$\\

\smallskip
No data
\end{minipage}
&\begin{minipage}[t]{4.5cm}
\vspace{3mm}
{ \normalsize ARL on NSE}
\\

\smallskip
$N=0$, 
$\overline{a}^+=0.0000$\\

\smallskip
No data
\end{minipage}
&\begin{minipage}[t]{4.5cm}
\vspace{3mm}
{ \normalsize ARL on Chicago}
\\

\smallskip
$N=0$, 
$\overline{a}^+=0.0000$\\

\smallskip
No data
\end{minipage}
\end{tabular}\\
\begin{tabular}{c|c|c} \hline 
\begin{minipage}[t]{4.5cm}
\vspace{3mm}
{ \normalsize ARL on NYSE}
\\

\smallskip
$N=5000$, 
$\overline{a}^+=7.3767$\\

\smallskip
Optimization failed
\end{minipage}
&
\begin{minipage}[t]{4.5cm}
\vspace{3mm}
{ \normalsize ARL on ARCA}
\\

\smallskip
$N=1000$, 
$\overline{a}^+=38.241818$\\

\smallskip
Optimization failed
\end{minipage}
&
\begin{minipage}[t]{4.5cm}
\vspace{3mm}
{ \normalsize ARL on NASDAQ T}
\\

\smallskip
$N=3639$, 
$\overline{a}^+=5051.364227$\\

\smallskip
$P_S=0.083$ \\

 \smallskip
$\kappa_0=0.0000(0.0000)^{***}$\\$\rho_0=0.0000(0.0000)^{***}$\\\smallskip
\end{minipage}
\end{tabular}\\
\begin{tabular}{c|c|c} \hline 
\begin{minipage}[t]{4.5cm}
\vspace{3mm}
{ \normalsize ARL on CBOE}
\\

\smallskip
$N=283$, 
$\overline{a}^+=501008.022436$\\

\smallskip
Optimization failed
\end{minipage}
&
\begin{minipage}[t]{4.5cm}
\vspace{3mm}
{ \normalsize ARL on BATS}
\\

\smallskip
$N=97$, 
$\overline{a}^+=9354.074766$\\

\smallskip
Optimization failed
\end{minipage}
&
\begin{minipage}[t]{4.5cm}
\vspace{3mm}
{ \normalsize ARL on ISE}
\\

\smallskip
$N=283$, 
$\overline{a}^+=72545.182254$\\

\smallskip
Optimization failed
\end{minipage}
\end{tabular}

\subsection{GZI Model}
\label{gzires}

\tiny
\begin{tabular}{c|c|c} \hline 
\begin{minipage}[t]{4.5cm}
\vspace{3mm}
{ \normalsize XOM on NASDAQ OB}
\\

\smallskip
$N=5000$, 
$\overline{a}^+=1.928909$, $\overline q=1.005368$\\

\smallskip
$P_S=0.225$ $P_m=0.194$
\\

 \smallskip
$\kappa_0=17.2408(1.3449)^{***}$\\$\kappa_1=0.7506(0.0773)^{***}$\\$\alpha_\kappa=-2.6031(0.1214)^{***}$\\$\rho_0=17.9548(1.3943)^{***}$\\$\rho_1=1.0378(0.1075)^{***}$\\$\alpha_\rho=-2.6008(0.1272)^{***}$\\$\eta=0.9230(0.0107)^{***}$\\\smallskip
\end{minipage}
&
\begin{minipage}[t]{4.5cm}
\vspace{3mm}
{ \normalsize XOM on NSE}
\\

\smallskip
$N=1000^\star$, 
$\overline{a}^+=36.201818$, $\overline q=1.433895$\\

\smallskip
$P_m=0.704$
\\

 \smallskip
$\kappa_0=0.3677(0.0160)^{***}$\\$\alpha_\kappa=-2.8310(0.0348)^{***}$\\$\rho_0=3.0327(0.0771)^{***}$\\$\alpha_\rho=-2.8919(0.0279)^{***}$\\$\eta=0.9845(0.0012)^{***}$\\\smallskip
\end{minipage}
&\begin{minipage}[t]{4.5cm}
\vspace{3mm}
{ \normalsize XOM on Chicago}
\\

\smallskip
$N=46$, 
$\overline{a}^+=7.4118$, $\overline q=2.3594$\\

\smallskip
$P_m=0.346$
\\

 \smallskip
$\kappa_0=0.0256(0.0051)^{***}$\\$\kappa_1=0.0210(0.0095)^{*}$\\$\alpha_\kappa=-2.2872(0.2175)^{***}$\\$\rho_0=0.0081(0.0016)^{***}$\\$\rho_1=0.2698(0.0721)^{***}$\\$\alpha_\rho=-3.2243(0.1371)^{***}$\\$\eta=0.9895(0.0118)^{***}$\\\smallskip
\end{minipage}
\end{tabular}\\
\begin{tabular}{c|c|c} \hline 
\begin{minipage}[t]{4.5cm}
\vspace{3mm}
{ \normalsize XOM on NYSE}
\\

\smallskip
$N=5000$, 
$\overline{a}^+=1.8062$, $\overline q=1.9048$\\

\smallskip
$P_S=0.168$ \\

 \smallskip
$\kappa_0=2.1825(0.0094)^{***}$\\$\rho_0=0.4998(0.0021)^{***}$\\$\eta=0.8044(0.0012)^{***}$\\\smallskip
\end{minipage}
&\begin{minipage}[t]{4.5cm}
\vspace{3mm}
{ \normalsize XOM on ARCA}
\\

\smallskip
$N=5000$, 
$\overline{a}^+=1.0658$, $\overline q=1.2073$\\

\smallskip
$P_S=0.838$ $P_m=0.837$
\\

 \smallskip
$\kappa_0=40.9699(0.6556)^{***}$\\$\kappa_1=44.1512(2.9864)^{***}$\\$\alpha_\kappa=-1.7462(0.1082)^{***}$\\$\rho_0=5.4837(0.0888)^{***}$\\$\rho_1=7.7849(0.5395)^{***}$\\$\alpha_\rho=-0.7655(0.1062)^{***}$\\$\eta=0.7761(0.0029)^{***}$\\\smallskip
\end{minipage}
&\begin{minipage}[t]{4.5cm}
\vspace{3mm}
{ \normalsize XOM on NASDAQ T}
\\

\smallskip
$N=5000$, 
$\overline{a}^+=1.0445$, $\overline q=1.2832$\\

\smallskip
$P_S=0.738$ \\

 \smallskip
$\kappa_0=128.2530(2.4351)^{***}$\\$\rho_0=18.3973(0.3463)^{***}$\\$\eta=0.7620(0.0057)^{***}$\\\smallskip
\end{minipage}
\end{tabular}\\
\begin{tabular}{c|c|c} \hline 
\begin{minipage}[t]{4.5cm}
\vspace{3mm}
{ \normalsize XOM on CBOE}
\\

\smallskip
$N=662$, 
$\overline{a}^+=65.5720$, $\overline q=1.4527$\\

\smallskip
$P_S=-0.220$ $P_m=-0.054$
\\

 \smallskip
$\kappa_0=0.2861(0.0419)^{***}$\\$\kappa_1=0.0676(0.0100)^{***}$\\$\alpha_\kappa=-1.6968(0.0410)^{***}$\\$\rho_0=4.4734(0.5142)^{***}$\\$\rho_1=6.9722(0.9634)^{***}$\\$\alpha_\rho=-2.2068(0.0468)^{***}$\\$\eta=0.9361(0.0033)^{***}$\\\smallskip
\end{minipage}
&\begin{minipage}[t]{4.5cm}
\vspace{3mm}
{ \normalsize XOM on BATS}
\\

\smallskip
$N=5000$, 
$\overline{a}^+=1.0104$, $\overline q=1.1257$\\

\smallskip
$P_S=0.335$ $P_m=0.313$
\\

 \smallskip
$\kappa_0=133.1251(2.4196)^{***}$\\$\alpha_\kappa=-4.8374(0.1596)^{***}$\\$\rho_0=16.8156(0.3043)^{***}$\\$\alpha_\rho=-4.5862(0.1486)^{***}$\\$\eta=0.7135(0.0038)^{***}$\\\smallskip
\end{minipage}
&\begin{minipage}[t]{4.5cm}
\vspace{3mm}
{ \normalsize XOM on ISE}
\\

\smallskip
$N=1000^\star$, 
$\overline{a}^+=17.153636$, $\overline q=1.333631$\\

\smallskip
$P_S=0.453$ $P_m=0.439$
\\

 \smallskip
$\kappa_0=2.3030(0.0567)^{***}$\\$\alpha_\kappa=-2.7609(0.0150)^{***}$\\$\rho_0=6.1707(0.0873)^{***}$\\$\alpha_\rho=-2.4294(0.0119)^{***}$\\$\eta=0.7971(0.0039)^{***}$\\\smallskip
\end{minipage}
\end{tabular}\\
\begin{tabular}{c|c|c} \hline 
\begin{minipage}[t]{4.5cm}
\vspace{3mm}
{ \normalsize MSFT on NASDAQ OB}
\\

\smallskip
$N=5000$, 
$\overline{a}^+=1.0331$, $\overline q=4.1941$\\

\smallskip
$P_S=0.671$
$P_m=0.679$
\\

 \smallskip
$\kappa_0=0.6863(0.0039)^{***}$\\$\alpha_\kappa=-0.4756(0.0199)^{***}$\\$\rho_0=0.0487(0.0003)^{***}$\\$\alpha_\rho=-0.6009(0.0192)^{***}$\\$\eta=0.8249(0.0006)^{***}$\\\smallskip
\end{minipage}
&\begin{minipage}[t]{4.5cm}
\vspace{3mm}
{ \normalsize MSFT on NSE}
\\

\smallskip
$N=5000$, 
$\overline{a}^+=1.0675$, $\overline q=1.7971$\\

\smallskip
$P_S=0.433$ $P_m=0.381$
\\

 \smallskip
$\kappa_0=0.0099(0.0003)^{***}$\\$\kappa_1=0.2286(0.0105)^{***}$\\$\alpha_\kappa=-2.6446(0.0807)^{***}$\\$\rho_0=0.0079(0.0001)^{***}$\\$\rho_1=0.1003(0.0040)^{***}$\\$\alpha_\rho=-2.3066(0.0740)^{***}$\\$\eta=0.9666(0.0011)^{***}$\\\smallskip
\end{minipage}
&\begin{minipage}[t]{4.5cm}
\vspace{3mm}
{ \normalsize MSFT on Chicago}
\\

\smallskip
$N=1092$, 
$\overline{a}^+=7.4010$, $\overline q=6.0947$\\

\smallskip
$P_S=0.543$ $P_m=0.490$
\\

 \smallskip
$\kappa_0=0.0020(0.0001)^{***}$\\$\kappa_1=0.0022(0.0000)^{***}$\\$\alpha_\kappa=-0.7520(0.0076)^{***}$\\$\rho_0=0.0028(0.0000)^{***}$\\$\rho_1=0.0014(0.0000)^{***}$\\$\alpha_\rho=-1.9361(0.0256)^{***}$\\$\eta=0.9214(0.0006)^{***}$\\\smallskip
\end{minipage}
\end{tabular}\\
\begin{tabular}{c|c|c} \hline 
\begin{minipage}[t]{4.5cm}
\vspace{3mm}
{ \normalsize MSFT on NYSE}
\\

\smallskip
$N=0$, 
$\overline{a}^+=0.0000$, $\overline q=0.0000$\\

\smallskip
No data.
\end{minipage}
&\begin{minipage}[t]{4.5cm}
\vspace{3mm}
{ \normalsize MSFT on ARCA}
\\

\smallskip
$N=5000$, 
$\overline{a}^+=1.0005$, $\overline q=2.9264$\\

\smallskip
Timeout reached
\end{minipage}
&\begin{minipage}[t]{4.5cm}
\vspace{3mm}
{ \normalsize MSFT on NASDAQ T}
\\

\smallskip
$N=0$, 
$\overline{a}^+=0.0000$, $\overline q=0.0000$\\

\smallskip
No data
\end{minipage}
\end{tabular}\\
\begin{tabular}{c|c|c} \hline 
\begin{minipage}[t]{4.5cm}
\vspace{3mm}
{ \normalsize MSFT on CBOE}
\\

\smallskip
$N=296$, 
$\overline{a}^+=315.9969$, $\overline q=1.5671$\\

\smallskip
$P_S=0.683$ $P_m=0.718$
\\

 \smallskip
$\kappa_0=0.0090(0.0012)^{***}$\\$\kappa_1=0.0091(0.0007)^{***}$\\$\alpha_\kappa=-0.3375(0.0483)^{***}$\\$\rho_0=0.0234(0.0022)^{***}$\\$\rho_1=0.0253(0.0023)^{***}$\\$\alpha_\rho=-9.2445(0.0873)^{***}$\\$\eta=0.8491(0.0034)^{***}$\\\smallskip
\end{minipage}
&
\begin{minipage}[t]{4.5cm}
\vspace{3mm}
{ \normalsize MSFT on BATS}
\\

\smallskip
$N=1000^\star$, 
$\overline{a}^+=1.000000$, $\overline q=1.637273$\\

\smallskip
$P_S=1.00$ \\

 \smallskip
$\kappa_0=3.0505(0.0067)^{***}$\\$\rho_0=0.0407(0.0001)^{***}$\\$\eta=0.9397(0.0004)^{***}$\\\smallskip
\end{minipage}
&
\begin{minipage}[t]{4.5cm}
\vspace{3mm}
{ \normalsize MSFT on ISE}
\\

\smallskip
$N=1000^\star$, 
$\overline{a}^+=1.503636$, $\overline q=4.573843$\\

\smallskip
$P_S=0.959$ $P_m=0.959$
\\

 \smallskip
$\kappa_0=899.9914(6.3647)^{***}$\\$\alpha_\kappa=-6.7585(0.0391)^{***}$\\$\rho_0=17.2797(0.1224)^{***}$\\$\alpha_\rho=-3.5550(0.0268)^{***}$\\$\eta=1.0000(0.0031)^{***}$\\\smallskip
\end{minipage}
\end{tabular}\\
\begin{tabular}{c|c|c} \hline 
\begin{minipage}[t]{4.5cm}
\vspace{3mm}
{ \normalsize GE on NASDAQ OB}
\\

\smallskip
$N=5000$, 
$\overline{a}^+=1.0480$, $\overline q=3.5815$\\

\smallskip
$P_S=0.736$ $P_m=0.737$
\\

 \smallskip
$\kappa_0=10.3450(0.0536)^{***}$\\$\alpha_\kappa=-4.1472(0.0266)^{***}$\\$\rho_0=0.7183(0.0035)^{***}$\\$\alpha_\rho=-4.2231(0.0290)^{***}$\\$\eta=0.8351(0.0005)^{***}$\\\smallskip
\end{minipage}
&\begin{minipage}[t]{4.5cm}
\vspace{3mm}
{ \normalsize GE on NSE}
\\

\smallskip
$N=5000$, 
$\overline{a}^+=1.0847$, $\overline q=9.0287$\\

\smallskip
$P_S=0.329$
\\

 \smallskip
$\kappa_0=1.3172(0.0026)^{***}$\\$\rho_0=0.0543(0.0001)^{***}$\\$\eta=0.9654(0.0003)^{***}$\\\smallskip
\end{minipage}
&\begin{minipage}[t]{4.5cm}
\vspace{3mm}
{ \normalsize GE on Chicago}
\\

\smallskip
$N=1311$, 
$\overline{a}^+=2.4913$, $\overline q=6.9945$\\

\smallskip
$P_S=0.790$
$P_m=0.790$
\\

 \smallskip
$\kappa_0=0.4598(0.0024)^{***}$\\$\alpha_\kappa=-1.4734(0.0039)^{***}$\\$\rho_0=0.0184(0.0001)^{***}$\\$\alpha_\rho=-1.7745(0.0059)^{***}$\\$\eta=0.7351(0.0009)^{***}$\\\smallskip
\end{minipage}
\end{tabular}\\
\begin{tabular}{c|c|c} \hline 
\begin{minipage}[t]{4.5cm}
\vspace{3mm}
{ \normalsize GE on NYSE}
\\

\smallskip
$N=1000^\star$, 
$\overline{a}^+=1.009091$, $\overline q=12.854730$\\

\smallskip
$P_S=1.000$ \\

 \smallskip
$\kappa_0=40.2572(0.1865)^{***}$\\$\rho_0=0.6672(0.0030)^{***}$\\$\eta=0.9157(0.0006)^{***}$\\\smallskip
\end{minipage}
&\begin{minipage}[t]{4.5cm}
\vspace{3mm}
{ \normalsize GE on ARCA}
\\

\smallskip
$N=5000$, 
$\overline{a}^+=1.0024$, $\overline q=4.7237$\\

\smallskip
Timeout reached
\end{minipage}
&\begin{minipage}[t]{4.5cm}
\vspace{3mm}
{ \normalsize GE on NASDAQ T}
\\

\smallskip
$N=5000$, 
$\overline{a}^+=1.0007$, $\overline q=6.1216$\\

\smallskip
Timeout reached
\end{minipage}
\end{tabular}\\
\begin{tabular}{c|c|c} \hline 
\begin{minipage}[t]{4.5cm}
\vspace{3mm}
{ \normalsize GE on CBOE}
\\

\smallskip
$N=413$, 
$\overline{a}^+=2.9824$, $\overline q=11.6940$\\

\smallskip
$P_S=0.218$ $P_m=0.227$
\\

 \smallskip
$\kappa_0=1.6704(0.0167)^{***}$\\$\kappa_1=4.1414(0.2666)^{***}$\\$\alpha_\kappa=-2.1172(0.0500)^{***}$\\$\rho_0=0.1456(0.0018)^{***}$\\$\rho_1=9.4944(0.5865)^{***}$\\$\alpha_\rho=-3.3851(0.0461)^{***}$\\$\eta=0.8420(0.0028)^{***}$\\\smallskip
\end{minipage}
&

\begin{minipage}[t]{4.5cm}
\vspace{3mm}
{ \normalsize GE on BATS}
\\

\smallskip
$N=1000^\star$, 
$\overline{a}^+=1.000000$, $\overline q=2.263397$\\

\smallskip
$P_S=1.00$
\\

 \smallskip
$\kappa_0=2.1701(0.0040)^{***}$\\$\rho_0=0.0373(0.0000)^{***}$\\$\eta=0.9770(0.0003)^{***}$\\\smallskip
\end{minipage}

&

\begin{minipage}[t]{4.5cm}
\vspace{3mm}
{ \normalsize GE on ISE}
\\

\smallskip
$N=1000^\star$, 
$\overline{a}^+=1.456364$, $\overline q=6.700976$\\

\smallskip
$P_S=0.360$ \\

 \smallskip
$\kappa_0=0.4013(0.0011)^{***}$\\$\rho_0=0.0018(0.0000)^{***}$\\$\eta=0.6216(0.0003)^{***}$\\\smallskip
\end{minipage}

\end{tabular}\\

\begin{tabular}{c|c|c} \hline 
\begin{minipage}[t]{4.5cm}
\vspace{3mm}
{ \normalsize MKTX on NASDAQ OMX BX}
\\

\smallskip
$N=1$, 
$\overline{a}^+=61.500000$, $\overline q=1.000000$\\

\smallskip
$P_m=-0.007$
\\

 \smallskip
$\kappa_0=0.1442(0.0000)^{***}$\\$\alpha_\kappa=0.0000(-nan)^{}$\\$\rho_0=168187.2977(-nan)^{}$\\$\alpha_\rho=-9.6653(-nan)^{}$\\$\eta=1.0000(-nan)^{}$\\\smallskip
\end{minipage}
&\begin{minipage}[t]{4.5cm}
\vspace{3mm}
{ \normalsize MKTX on NSE}
\\

\smallskip
$N=4$, 
$\overline{a}^+=6.0000$, $\overline q=1.0000$\\

\smallskip
No data
\end{minipage}
&\begin{minipage}[t]{4.5cm}
\vspace{3mm}
{ \normalsize MKTX on Chicago}
\\

\smallskip
$N=0$, 
$\overline{a}^+=0.0000$, $\overline q=0.0000$\\

\smallskip
No data
\end{minipage}
\end{tabular}\\
\begin{tabular}{c|c|c} \hline 
\begin{minipage}[t]{4.5cm}
\vspace{3mm}
{ \normalsize MKTX on NYSE}
\\

\smallskip
$N=0$, 
$\overline{a}^+=0.0000$, $\overline q=0.0000$\\

\smallskip
No data
\end{minipage}
&\begin{minipage}[t]{4.5cm}
\vspace{3mm}
{ \normalsize MKTX on ARCA}
\\

\smallskip
$N=4379$, 
$\overline{a}^+=5.5327$, $\overline q=1.2060$\\

\smallskip
$P_S=0.314$ $P_m=0.311$
\\

 \smallskip
$\kappa_0=0.0042(0.0001)^{***}$\\$\kappa_1=0.0027(0.0001)^{***}$\\$\alpha_\kappa=-0.0393(0.0096)^{***}$\\$\rho_0=0.0042(0.0001)^{***}$\\$\rho_1=0.0110(0.0001)^{***}$\\$\alpha_\rho=-0.4967(0.0067)^{***}$\\$\eta=0.6193(0.0016)^{***}$\\\smallskip
\end{minipage}
&\begin{minipage}[t]{4.5cm}
\vspace{3mm}
{ \normalsize MKTX on NASDAQ T}
\\

\smallskip
$N=0$, 
$\overline{a}^+=0.0000$, $\overline q=0.0000$\\

\smallskip
No data
\end{minipage}
\end{tabular}\\
\begin{tabular}{c|c|c} \hline 
\begin{minipage}[t]{4.5cm}
\vspace{3mm}
{ \normalsize MKTX on CBOE}
\\

\smallskip
$N=100$, 
$\overline{a}^+=1049.7568$, $\overline q=1.3153$\\

\smallskip
$P_S=0.998$ \\

 \smallskip
$\kappa_0=0.0000(0.0000)^{***}$\\$\rho_0=0.0012(0.0004)^{***}$\\$\eta=1.0000(0.0138)^{***}$\\\smallskip
\end{minipage}
&\begin{minipage}[t]{4.5cm}
\vspace{3mm}
{ \normalsize MKTX on BATS}
\\

\smallskip
$N=1623$, 
$\overline{a}^+=18.0118$, $\overline q=1.1758$\\

\smallskip
$P_S=0.411$ \\

 \smallskip
$\kappa_0=0.0004(0.0000)^{***}$\\$\rho_0=0.0005(0.0000)^{***}$\\$\eta=0.8852(0.0017)^{***}$\\\smallskip
\end{minipage}
&\begin{minipage}[t]{4.5cm}
\vspace{3mm}
{ \normalsize MKTX on ISE}
\\

\smallskip
$N=590$, 
$\overline{a}^+=26.6271$, $\overline q=1.2500$\\

\smallskip
$P_S=-1.121$ \\

 \smallskip
$\kappa_0=0.0010(0.0000)^{***}$\\$\rho_0=0.0011(0.0000)^{***}$\\$\eta=0.3867(0.0031)^{***}$\\\smallskip
\end{minipage}
\end{tabular}\\
\begin{tabular}{c|c|c} \hline 
\begin{minipage}[t]{4.5cm}
\vspace{3mm}
{ \normalsize JCOM on NASDAQ OMX BX}
\\

\smallskip
$N=1900$, 
$\overline{a}^+=51.474163$, $\overline q=1.203507$\\

\smallskip
$P_S=0.904$ $P_m=0.772$
\\

 \smallskip
$\kappa_0=0.0050(0.0004)^{***}$\\$\kappa_1=0.0014(0.0002)^{***}$\\$\alpha_\kappa=-1.8681(0.1083)^{***}$\\$\rho_0=0.0193(0.0007)^{***}$\\$\rho_1=0.0256(0.0028)^{***}$\\$\alpha_\rho=-2.2825(0.1025)^{***}$\\$\eta=0.9998(0.0001)^{***}$\\\smallskip
\end{minipage}
&\begin{minipage}[t]{4.5cm}
\vspace{3mm}
{ \normalsize JCOM on NSE}
\\

\smallskip
$N=13$, 
$\overline{a}^+=8.2667$, $\overline q=1.0769$\\

\smallskip
$P_S=-0.439$ $P_m=-0.400$
\\

 \smallskip
$\kappa_0=2.0000(nan)^{}$\\$\kappa_1=2.0000(nan)^{}$\\$\alpha_\kappa=0.0000(nan)^{}$\\$\rho_0=2.0000(nan)^{}$\\$\rho_1=1.0000(nan)^{}$\\$\alpha_\rho=0.0000(nan)^{}$\\$\eta=1.0000(nan)^{}$\\\smallskip
\end{minipage}
&\begin{minipage}[t]{4.5cm}
\vspace{3mm}
{ \normalsize JCOM on Chicago}
\\

\smallskip
$N=0$, 
$\overline{a}^+=0.0000$, $\overline q=0.0000$\\

\smallskip
No data
\end{minipage}
\end{tabular}\\
\begin{tabular}{c|c|c} \hline 
\begin{minipage}[t]{4.5cm}
\vspace{3mm}
{ \normalsize JCOM on NYSE}
\\

\smallskip
$N=0$, 
$\overline{a}^+=0.0000$, $\overline q=0.0000$\\

\smallskip
No data
\end{minipage}
&\begin{minipage}[t]{4.5cm}
\vspace{3mm}
{ \normalsize JCOM on ARCA}
\\

\smallskip
$N=5000$, 
$\overline{a}^+=3.0464$, $\overline q=1.1506$\\

\smallskip
$P_S=0.082$ \\

 \smallskip
$\kappa_0=0.1442(0.0027)^{***}$\\$\rho_0=0.1998(0.0037)^{***}$\\$\eta=0.5970(0.0062)^{***}$\\\smallskip
\end{minipage}
&\begin{minipage}[t]{4.5cm}
\vspace{3mm}
{ \normalsize JCOM on NASDAQ T}
\\

\smallskip
$N=0$, 
$\overline{a}^+=0.0000$, $\overline q=0.0000$\\

\smallskip
No data
\end{minipage}
\end{tabular}\\
\begin{tabular}{c|c|c} \hline 
\begin{minipage}[t]{4.5cm}
\vspace{3mm}
{ \normalsize JCOM on CBOE}
\\

\smallskip
$N=118$, 
$\overline{a}^+=42.8000$, $\overline q=1.1923$\\

\smallskip
$P_S=0.295$ $P_m=0.291$
\\

 \smallskip
$\kappa_0=0.0271(0.0029)^{***}$\\$\alpha_\kappa=-0.4791(0.0330)^{***}$\\$\rho_0=0.0380(0.0045)^{***}$\\$\alpha_\rho=-0.3331(0.0354)^{***}$\\$\eta=0.9490(0.0051)^{***}$\\\smallskip
\end{minipage}
&\begin{minipage}[t]{4.5cm}
\vspace{3mm}
{ \normalsize JCOM on BATS}
\\

\smallskip
$N=5000$, 
$\overline{a}^+=15.0924$, $\overline q=1.1036$\\

\smallskip
$P_S=0.690$ $P_m=0.653$
\\

 \smallskip
$\kappa_0=0.0044(0.0002)^{***}$\\$\kappa_1=0.0102(0.0003)^{***}$\\$\alpha_\kappa=-0.9379(0.0086)^{***}$\\$\rho_0=0.0041(0.0001)^{***}$\\$\rho_1=0.0635(0.0013)^{***}$\\$\alpha_\rho=-1.4285(0.0100)^{***}$\\$\eta=0.7080(0.0022)^{***}$\\\smallskip
\end{minipage}
&\begin{minipage}[t]{4.5cm}
\vspace{3mm}
{ \normalsize JCOM on ISE}
\\

\smallskip
$N=2224$, 
$\overline{a}^+=46.0266$, $\overline q=1.1379$\\

\smallskip
$P_S=-0.839$ \\

 \smallskip
$\kappa_0=0.0000(0.0000)^{***}$\\$\rho_0=0.0000(0.0000)^{***}$\\$\eta=0.4636(0.0011)^{***}$\\\smallskip
\end{minipage}
\end{tabular}\\
\begin{tabular}{c|c|c} \hline 
\begin{minipage}[t]{4.5cm}
\vspace{3mm}
{ \normalsize ARL on NASDAQ OMX BX}
\\

\smallskip
$N=0$, 
$\overline{a}^+=0.000000$, $\overline q=0.000000$\\

\smallskip
No data
\end{minipage}
&\begin{minipage}[t]{4.5cm}
\vspace{3mm}
{ \normalsize ARL on NSE}
\\

\smallskip
$N=0$, 
$\overline{a}^+=0.000000$, $\overline q=0.000000$\\

\smallskip
No data
\end{minipage}
&\begin{minipage}[t]{4.5cm}
\vspace{3mm}
{ \normalsize ARL on Chicago}
\\

\smallskip
$N=0$, 
$\overline{a}^+=0.000000$, $\overline q=0.000000$\\

\smallskip
No data
\end{minipage}
\end{tabular}\\
\begin{tabular}{c|c|c} \hline 
\begin{minipage}[t]{4.5cm}
\vspace{3mm}
{ \normalsize ARL on NYSE}
\\

\smallskip
$N=299$, 
$\overline{a}^+=7.793313$, $\overline q=1.310448$\\

\smallskip
$P_S=-7.992$ \\

 \smallskip
$\kappa_0=0.0000(0.0000)^{***}$\\$\rho_0=0.0001(0.0000)^{***}$\\$\eta=0.8833(0.0089)^{***}$\\\smallskip
\end{minipage}
&\begin{minipage}[t]{4.5cm}
\vspace{3mm}
{ \normalsize ARL on ARCA}
\\

\smallskip
$N=78$, 
$\overline{a}^+=16204.7558$, $\overline q=1.5889$\\

\smallskip
Optimization failed
\end{minipage}
&\begin{minipage}[t]{4.5cm}
\vspace{3mm}
{ \normalsize ARL on NASDAQ T}
\\

\smallskip
$N=90$, 
$\overline{a}^+=17.1500$, $\overline q=1.1569$\\

\smallskip
$P_S=-0.246$ \\

 \smallskip
$\kappa_0=0.0000(0.0000)^{***}$\\$\rho_0=4.8238(0.4067)^{***}$\\$\eta=1.0000(0.0494)^{***}$\\\smallskip
\end{minipage}
\end{tabular}\\
\begin{tabular}{c|c|c} \hline 
\begin{minipage}[t]{4.5cm}
\vspace{3mm}
{ \normalsize ARL on CBOE}
\\

\smallskip
$N=0$, 
$\overline{a}^+=150.0000$, $\overline q=1.0000$\\

\smallskip
No data
\end{minipage}
&\begin{minipage}[t]{4.5cm}
\vspace{3mm}
{ \normalsize ARL on BATS}
\\

\smallskip
$N=7$, 
$\overline{a}^+=78.8750$, $\overline q=1.1250$\\

\smallskip
$P_S=0.997$ \\

 \smallskip
$\kappa_0=0.0219(0.0019)^{***}$\\$\rho_0=0.1378(0.0279)^{***}$\\$\eta=1.0000(0.0789)^{***}$\\\smallskip
\end{minipage}
&\begin{minipage}[t]{4.5cm}
\vspace{3mm}
{ \normalsize ARL on ISE}
\\

\smallskip
$N=6$, 
$\overline{a}^+=28469.4286$, $\overline q=1.8571$\\

\smallskip
$P_S=0.994$ $P_m=0.994$
\\

 \smallskip
$\kappa_0=2.0000(nan)^{}$\\$\kappa_1=2.0000(nan)^{}$\\$\alpha_\kappa=0.0000(nan)^{}$\\$\rho_0=2.0000(nan)^{}$\\$\rho_1=1.0000(nan)^{}$\\$\alpha_\rho=0.0000(nan)^{}$\\$\eta=1.0000(nan)^{}$\\\smallskip
\end{minipage}
\end{tabular}\\
\begin{tabular}{c|c|c} \hline 
\begin{minipage}[t]{4.5cm}
\vspace{3mm}
{ \normalsize ARL on NASDAQ OMX BX}
\\

\smallskip
$N=0$, 
$\overline{a}^+=0.000000$, $\overline q=0.000000$\\

\smallskip
No data
\end{minipage}
&\begin{minipage}[t]{4.5cm}
\vspace{3mm}
{ \normalsize ARL on NSE}
\\

\smallskip
$N=0$, 
$\overline{a}^+=0.000000$, $\overline q=0.000000$\\

\smallskip
No data
\end{minipage}
&\begin{minipage}[t]{4.5cm}
\vspace{3mm}
{ \normalsize ARL on Chicago}
\\

\smallskip
$N=0$, 
$\overline{a}^+=0.000000$, $\overline q=0.000000$\\

\smallskip
No data
\end{minipage}
\end{tabular}\\
\begin{tabular}{c|c|c} \hline 
\begin{minipage}[t]{4.5cm}
\vspace{3mm}
{ \normalsize ARL on NYSE}
\\

\smallskip
$N=299$, 
$\overline{a}^+=7.793313$, $\overline q=1.310448$\\

\smallskip
$P_S=-7.992$ \\

 \smallskip
$\kappa_0=0.0000(0.0000)^{***}$\\$\rho_0=0.0001(0.0000)^{***}$\\$\eta=0.8833(0.0089)^{***}$\\\smallskip
\end{minipage}
&\begin{minipage}[t]{4.5cm}
\vspace{3mm}
{ \normalsize ARL on ARCA}
\\

\smallskip
$N=78$, 
$\overline{a}^+=16204.7558$, $\overline q=1.5889$\\

\smallskip
$P_S=0.998$ $P_m=0.998$
\\

 \smallskip
$\kappa_0=2.0000(nan)^{}$\\$\kappa_1=2.0000(nan)^{}$\\$\alpha_\kappa=0.0000(nan)^{}$\\$\rho_0=2.0000(nan)^{}$\\$\rho_1=1.0000(nan)^{}$\\$\alpha_\rho=0.0000(nan)^{}$\\$\eta=1.0000(nan)^{}$\\\smallskip
\end{minipage}
&\begin{minipage}[t]{4.5cm}
\vspace{3mm}
{ \normalsize ARL on NASDAQ T}
\\

\smallskip
$N=90$, 
$\overline{a}^+=17.1500$, $\overline q=1.1569$\\

\smallskip
$P_S=-0.246$ \\

 \smallskip
$\kappa_0=0.0000(0.0000)^{***}$\\$\rho_0=4.8238(0.4067)^{***}$\\$\eta=1.0000(0.0494)^{***}$\\\smallskip
\end{minipage}
\end{tabular}\\
\begin{tabular}{c|c|c} \hline 
\begin{minipage}[t]{4.5cm}
\vspace{3mm}
{ \normalsize ARL on CBOE}
\\

\smallskip
$N=0$, 
$\overline{a}^+=150.0000$, $\overline q=1.0000$\\

\smallskip
No data
\end{minipage}
&\begin{minipage}[t]{4.5cm}
\vspace{3mm}
{ \normalsize ARL on BATS}
\\

\smallskip
$N=7$, 
$\overline{a}^+=78.8750$, $\overline q=1.1250$\\

\smallskip
$P_S=0.997$ \\

 \smallskip
$\kappa_0=0.0219(0.0019)^{***}$\\$\rho_0=0.1378(0.0279)^{***}$\\$\eta=1.0000(0.0789)^{***}$\\\smallskip
\end{minipage}
&\begin{minipage}[t]{4.5cm}
\vspace{3mm}
{ \normalsize ARL on ISE}
\\

\smallskip
$N=6$, 
$\overline{a}^+=28469.4286$, $\overline q=1.8571$\\

\smallskip
$P_S=0.994$ $P_m=0.994$
\\

 \smallskip
$\kappa_0=2.0000(nan)^{}$\\$\kappa_1=2.0000(nan)^{}$\\$\alpha_\kappa=0.0000(nan)^{}$\\$\rho_0=2.0000(nan)^{}$\\$\rho_1=1.0000(nan)^{}$\\$\alpha_\rho=0.0000(nan)^{}$\\$\eta=1.0000(nan)^{}$\\\smallskip
\end{minipage}
\end{tabular}

\subsection{ZI Model with Trades}
\label{gzitres}

\tiny
\begin{tabular}{c|c|c} \hline 
\begin{minipage}[t]{4.5cm}
\vspace{3mm}
{ \normalsize XOM on NASDAQ OB}
\\

\smallskip
$N=5000$, 
$\overline{a}^+=1.928909$\\

\smallskip
$P_S=0.183$ $P_m=0.170$
\\

 \smallskip
$\kappa_0=42.3316(4.0871)^{***}$\\$\kappa_1=1.3307(0.1302)^{***}$\\$\alpha_\kappa=-2.8164(0.1107)^{***}$\\$\rho_0=44.2052(4.2340)^{***}$\\$\rho_1=1.8553(0.1773)^{***}$\\$\alpha_\rho=-2.6764(0.1054)^{***}$\\\smallskip
\end{minipage}
&
\begin{minipage}[t]{4.5cm}
\vspace{3mm}
{ \normalsize XOM on NSE}
\\

\smallskip
$N=1000^\star$, 
$\overline{a}^+=36.201818$\\

\smallskip
$P_S=0.649$ $P_m=0.695$
\\

 \smallskip
$\kappa_0=0.3855(0.0167)^{***}$\\$\alpha_\kappa=-2.7543(0.0340)^{***}$\\$\rho_0=3.0505(0.0657)^{***}$\\$\alpha_\rho=-2.6481(0.0156)^{***}$\\\smallskip
\end{minipage}
&\begin{minipage}[t]{4.5cm}
\vspace{3mm}
{ \normalsize XOM on Chicago}
\\

\smallskip
$N=46$, 
$\overline{a}^+=7.4118$\\

\smallskip
$P_m=0.340$
\\

 \smallskip
$\kappa_0=0.0265(0.0048)^{***}$\\$\kappa_1=0.0188(0.0079)^{**}$\\$\alpha_\kappa=-2.2116(0.2067)^{***}$\\$\rho_0=0.0084(0.0015)^{***}$\\$\rho_1=0.2206(0.0529)^{***}$\\$\alpha_\rho=-3.0384(0.1210)^{***}$\\\smallskip
\end{minipage}
\end{tabular}\\
\begin{tabular}{c|c|c} \hline 
\begin{minipage}[t]{4.5cm}
\vspace{3mm}
{ \normalsize XOM on NYSE}
\\

\smallskip
$N=5000$, 
$\overline{a}^+=1.8062$\\

\smallskip
$P_S=0.133$ \\

 \smallskip
$\kappa_0=3.5635(0.0135)^{***}$\\$\rho_0=0.9331(0.0030)^{***}$\\\smallskip
\end{minipage}
&\begin{minipage}[t]{4.5cm}
\vspace{3mm}
{ \normalsize XOM on ARCA}
\\

\smallskip
$N=5000$, 
$\overline{a}^+=1.0658$\\

\smallskip
$P_S=0.835$ $P_m=0.836$
\\

 \smallskip
$\kappa_0=122.6212(1.3671)^{***}$\\$\kappa_1=48.9888(3.6449)^{***}$\\$\alpha_\kappa=-1.5287(0.0914)^{***}$\\$\rho_0=17.0316(0.1812)^{***}$\\$\rho_1=9.6314(0.7186)^{***}$\\$\alpha_\rho=-0.5917(0.0862)^{***}$\\\smallskip
\end{minipage}
&\begin{minipage}[t]{4.5cm}
\vspace{3mm}
{ \normalsize XOM on NASDAQ T}
\\

\smallskip
$N=5000$, 
$\overline{a}^+=1.0445$\\

\smallskip
$P_S=0.739$ $P_m=0.739$
\\

 \smallskip
$\kappa_0=328.9379(5.7143)^{***}$\\$\alpha_\kappa=-1.0492(0.2231)^{***}$\\$\rho_0=47.4101(0.8032)^{***}$\\$\alpha_\rho=-1.0763(0.2234)^{***}$\\\smallskip
\end{minipage}
\end{tabular}\\
\begin{tabular}{c|c|c} \hline 
\begin{minipage}[t]{4.5cm}
\vspace{3mm}
{ \normalsize XOM on CBOE}
\\

\smallskip
$N=662$, 
$\overline{a}^+=65.5720$\\

\smallskip
$P_S=-0.335$ $P_m=-0.069$
\\

 \smallskip
$\kappa_0=0.2892(0.0431)^{***}$\\$\kappa_1=0.0698(0.0109)^{***}$\\$\alpha_\kappa=-1.6006(0.0429)^{***}$\\$\rho_0=4.6186(0.5197)^{***}$\\$\rho_1=6.4110(0.8990)^{***}$\\$\alpha_\rho=-1.9299(0.0365)^{***}$\\\smallskip
\end{minipage}
&\begin{minipage}[t]{4.5cm}
\vspace{3mm}
{ \normalsize XOM on BATS}
\\

\smallskip
$N=5000$, 
$\overline{a}^+=1.0104$\\

\smallskip
$P_S=0.599$ $P_m=0.598$
\\

 \smallskip
$\kappa_0=4547.3608(195.8280)^{***}$\\$\alpha_\kappa=-3.3412(0.3168)^{***}$\\$\rho_0=584.8992(25.1806)^{***}$\\$\alpha_\rho=-2.9657(0.3183)^{***}$\\\smallskip
\end{minipage}
&
\begin{minipage}[t]{4.5cm}
\vspace{3mm}
{ \normalsize XOM on ISE}
\\

\smallskip
$N=1000^\star$, 
$\overline{a}^+=17.153636$\\

\smallskip
$P_S=0.417$ \\

 \smallskip
$\kappa_0=0.0529(0.0006)^{***}$\\$\rho_0=0.3351(0.0020)^{***}$\\\smallskip
\end{minipage}
\end{tabular}\\
\begin{tabular}{c|c|c} \hline 
\begin{minipage}[t]{4.5cm}
\vspace{3mm}
{ \normalsize MSFT on NASDAQ OB}
\\

\smallskip
$N=5000$, 
$\overline{a}^+=1.0331$\\

\smallskip
$P_S=0.665$ $P_m=0.664$
\\

 \smallskip
$\kappa_0=3.2470(0.0128)^{***}$\\$\alpha_\kappa=-0.0807(0.0286)^{**}$\\$\rho_0=0.2541(0.0008)^{***}$\\$\alpha_\rho=-0.3001(0.0270)^{***}$\\\smallskip
\end{minipage}
&\begin{minipage}[t]{4.5cm}
\vspace{3mm}
{ \normalsize MSFT on NSE}
\\

\smallskip
$N=5000$, 
$\overline{a}^+=1.0675$\\

\smallskip
$P_S=0.303$ $P_m=0.288$
\\

 \smallskip
$\kappa_0=0.0167(0.0005)^{***}$\\$\kappa_1=0.1939(0.0095)^{***}$\\$\alpha_\kappa=-2.4091(0.0680)^{***}$\\$\rho_0=0.0135(0.0001)^{***}$\\$\rho_1=0.0823(0.0037)^{***}$\\$\alpha_\rho=-2.0423(0.0669)^{***}$\\\smallskip
\end{minipage}
&\begin{minipage}[t]{4.5cm}
\vspace{3mm}
{ \normalsize MSFT on Chicago}
\\

\smallskip
$N=1092$, 
$\overline{a}^+=7.4010$\\

\smallskip
$P_S=0.463$ $P_m=0.512$
\\

 \smallskip
$\kappa_0=0.0012(0.0000)^{***}$\\$\alpha_\kappa=-0.3006(0.0072)^{***}$\\$\rho_0=0.0024(0.0000)^{***}$\\$\alpha_\rho=-0.9152(0.0049)^{***}$\\\smallskip
\end{minipage}
\end{tabular}\\
\begin{tabular}{c|c|c} \hline 
\begin{minipage}[t]{4.5cm}
\vspace{3mm}
{ \normalsize MSFT on NYSE}
\\

\smallskip
$N=0$, 
$\overline{a}^+=0.0000$\\

\smallskip
No data
\end{minipage}
&\begin{minipage}[t]{4.5cm}
\vspace{3mm}
{ \normalsize MSFT on ARCA}
\\

\smallskip
$N=1000^\star$, 
$\overline{a}^+=1.000000$\\

\smallskip
$P_S=1.0$
\\

 \smallskip
$\kappa_0=1.4199(0.0037)^{***}$\\$\rho_0=0.0124(0.0000)^{***}$\\\smallskip
\end{minipage}
&\begin{minipage}[t]{4.5cm}
\vspace{3mm}
{ \normalsize MSFT on NASDAQ T}
\\

\smallskip
$N=0$, 
$\overline{a}^+=0.0000$\\

\smallskip
No data
\end{minipage}
\end{tabular}\\
\begin{tabular}{c|c|c} \hline 
\begin{minipage}[t]{4.5cm}
\vspace{3mm}
{ \normalsize MSFT on CBOE}
\\

\smallskip
$N=296$, 
$\overline{a}^+=315.9969$\\

\smallskip
$P_S=0.680$ \\

 \smallskip
$\kappa_0=0.0146(0.0004)^{***}$\\$\rho_0=0.0199(0.0006)^{***}$\\\smallskip
\end{minipage} &
\begin{minipage}[t]{4.5cm}
\vspace{3mm}
{ \normalsize MSFT on BATS}
\\

\smallskip
$N=1000$, 
$\overline{a}^+=1.000000$\\

\smallskip
$P_S=1.0$ \\

 \smallskip
$\kappa_0=1664.8294(16.2923)^{***}$\\$\rho_0=16.1257(0.1567)^{***}$\\\smallskip
\end{minipage}
&
\begin{minipage}[t]{4.5cm}
\vspace{3mm}
{ \normalsize MSFT on ISE}
\\

\smallskip
$N=1000^\star$, 
$\overline{a}^+=1.503636$\\

\smallskip
$P_S=0.959$ $P_m=0.959$
\\

 \smallskip
$\kappa_0=899.9914(6.1338)^{***}$\\$\alpha_\kappa=-6.7585(0.0389)^{***}$\\$\rho_0=17.2797(0.1171)^{***}$\\$\alpha_\rho=-3.5550(0.0267)^{***}$\\\smallskip
\end{minipage}
\end{tabular}\\
\begin{tabular}{c|c|c} \hline 
\begin{minipage}[t]{4.5cm}
\vspace{3mm}
{ \normalsize GE on NASDAQ OB}
\\

\smallskip
$N=5000$, 
$\overline{a}^+=1.0480$\\

\smallskip
$P_S=0.742$ \\

 \smallskip
$\kappa_0=68.1041(0.2727)^{***}$\\$\rho_0=5.0341(0.0183)^{***}$\\\smallskip
\end{minipage}
&\begin{minipage}[t]{4.5cm}
\vspace{3mm}
{ \normalsize GE on NSE}
\\

\smallskip
$N=5000$, 
$\overline{a}^+=1.0847$\\

\smallskip
Timeout reached
\end{minipage}
&\begin{minipage}[t]{4.5cm}
\vspace{3mm}
{ \normalsize GE on Chicago}
\\

\smallskip
$N=1311$, 
$\overline{a}^+=2.4913$\\

\smallskip
$P_S=0.790$ $P_m=0.790$
\\

 \smallskip
$\kappa_0=3.4210(0.0154)^{***}$\\$\alpha_\kappa=-2.5542(0.0051)^{***}$\\$\rho_0=0.1435(0.0006)^{***}$\\$\alpha_\rho=-2.9388(0.0073)^{***}$\\\smallskip
\end{minipage}
\end{tabular}\\
\begin{tabular}{c|c|c} \hline 
\begin{minipage}[t]{4.5cm}
\vspace{3mm}
{ \normalsize GE on NYSE}
\\

\smallskip
$N=1000^\star$, 
$\overline{a}^+=1.009091$\\

\smallskip
$P_S=1.000$ \\

 \smallskip
$\kappa_0=67.4131(0.2748)^{***}$\\$\rho_0=1.1715(0.0046)^{***}$\\\smallskip
\end{minipage}
&\begin{minipage}[t]{4.5cm}
\vspace{3mm}
{ \normalsize GE on ARCA}
\\

\smallskip
$N=5000$, 
$\overline{a}^+=1.0024$\\

\smallskip
Timeout reached
\end{minipage}
&\begin{minipage}[t]{4.5cm}
\vspace{3mm}
{ \normalsize GE on NASDAQ T}
\\

\smallskip
$N=5000$, 
$\overline{a}^+=1.0007$\\

\smallskip
Timeout reached
\end{minipage}
\end{tabular}\\
\begin{tabular}{c|c|c} \hline 
\begin{minipage}[t]{4.5cm}
\vspace{3mm}
{ \normalsize GE on CBOE}
\\

\smallskip
$N=413$, 
$\overline{a}^+=2.9824$\\

\smallskip
$P_S=0.222$ $P_m=0.233$
\\

 \smallskip
$\kappa_0=1.9455(0.0178)^{***}$\\$\kappa_1=4.4616(0.2844)^{***}$\\$\alpha_\kappa=-1.8417(0.0465)^{***}$\\$\rho_0=0.1828(0.0017)^{***}$\\$\rho_1=8.6946(0.5419)^{***}$\\$\alpha_\rho=-2.8376(0.0460)^{***}$\\\smallskip
\end{minipage}
&
\begin{minipage}[t]{4.5cm}
\vspace{3mm}
{ \normalsize GE on BATS}
\\

\smallskip
$N=1000^\star$, 
$\overline{a}^+=1.000000$\\

\smallskip
$P_S=1.0$
\\

 \smallskip
$\kappa_0=1.4835(0.0031)^{***}$\\$\rho_0=0.0392(0.0000)^{***}$\\\smallskip
\end{minipage}
&
\begin{minipage}[t]{4.5cm}
\vspace{3mm}
{ \normalsize GE on ISE}
\\

\smallskip
$N=1000^\star$, 
$\overline{a}^+=1.456364$\\

\smallskip
$P_S=0.417$
\\

 \smallskip
$\kappa_0=1.1043(0.0026)^{***}$\\$\rho_0=0.0093(0.0000)^{***}$\\\smallskip
\end{minipage}
\end{tabular}\\

\begin{tabular}{c|c|c} \hline 
\begin{minipage}[t]{4.5cm}
\vspace{3mm}
{ \normalsize MKTX on NASDAQ OB}
\\

\smallskip
$N=1$, 
$\overline{a}^+=61.500000$\\

\smallskip
No data
\end{minipage}
&\begin{minipage}[t]{4.5cm}
\vspace{3mm}
{ \normalsize MKTX on NSE}
\\

\smallskip
$N=4$, 
$\overline{a}^+=6.000000$\\

\smallskip
No data
\end{minipage}
&\begin{minipage}[t]{4.5cm}
\vspace{3mm}
{ \normalsize MKTX on Chicago}
\\

\smallskip
$N=0$, 
$\overline{a}^+=0.0000$\\

\smallskip
No data
\end{minipage}
\end{tabular}\\
\begin{tabular}{c|c|c} \hline 
\begin{minipage}[t]{4.5cm}
\vspace{3mm}
{ \normalsize MKTX on NYSE}
\\

\smallskip
$N=0$, 
$\overline{a}^+=0.0000$\\

\smallskip
No data
\end{minipage}
&\begin{minipage}[t]{4.5cm}
\vspace{3mm}
{ \normalsize MKTX on ARCA}
\\

\smallskip
$N=4379$, 
$\overline{a}^+=5.5327$\\

\smallskip
$P_S=0.269$ $P_m=0.250$
\\

 \smallskip
$\kappa_0=0.0666(0.0010)^{***}$\\$\kappa_1=0.0167(0.0003)^{***}$\\$\alpha_\kappa=-0.4550(0.0080)^{***}$\\$\rho_0=0.1096(0.0011)^{***}$\\$\rho_1=0.0708(0.0006)^{***}$\\$\alpha_\rho=-0.8269(0.0048)^{***}$\\\smallskip
\end{minipage}
&\begin{minipage}[t]{4.5cm}
\vspace{3mm}
{ \normalsize MKTX on NASDAQ T}
\\

\smallskip
$N=0$, 
$\overline{a}^+=0.0000$\\

\smallskip
No data
\end{minipage}
\end{tabular}\\
\begin{tabular}{c|c|c} \hline 
\begin{minipage}[t]{4.5cm}
\vspace{3mm}
{ \normalsize MKTX on CBOE}
\\

\smallskip
$N=100$, 
$\overline{a}^+=1049.7568$\\

\smallskip
$P_S=0.996$ \\

 \smallskip
$\kappa_0=0.0000(0.0000)^{***}$\\$\rho_0=0.0018(0.0002)^{***}$\\\smallskip
\end{minipage}
&\begin{minipage}[t]{4.5cm}
\vspace{3mm}
{ \normalsize MKTX on BATS}
\\

\smallskip
$N=1623$, 
$\overline{a}^+=18.0118$\\

\smallskip
$P_S=0.394$ \\

 \smallskip
$\kappa_0=0.0005(0.0000)^{***}$\\$\rho_0=0.0012(0.0000)^{***}$\\\smallskip
\end{minipage}
&\begin{minipage}[t]{4.5cm}
\vspace{3mm}
{ \normalsize MKTX on ISE}
\\

\smallskip
$N=590$, 
$\overline{a}^+=26.6271$\\

\smallskip
$P_S=-1.942$ \\

 \smallskip
$\kappa_0=8.5760(0.2158)^{***}$\\$\rho_0=27.1661(0.6893)^{***}$\\\smallskip
\end{minipage}
\end{tabular}\\
\begin{tabular}{c|c|c} \hline 
\begin{minipage}[t]{4.5cm}
\vspace{3mm}
{ \normalsize JCOM on NASDAQ OB}
\\

\smallskip
$N=1900$, 
$\overline{a}^+=51.4742$\\

\smallskip
$P_S=0.903$ $P_m=0.789$
\\

 \smallskip
$\kappa_0=0.0054(0.0004)^{***}$\\$\alpha_\kappa=-2.5172(0.0804)^{***}$\\$\rho_0=0.0248(0.0007)^{***}$\\$\alpha_\rho=-2.0220(0.0440)^{***}$\\\smallskip
\end{minipage}
&\begin{minipage}[t]{4.5cm}
\vspace{3mm}
{ \normalsize JCOM on NSE}
\\

\smallskip
$N=13$, 
$\overline{a}^+=8.2667$\\

\smallskip
No data
\end{minipage}
&\begin{minipage}[t]{4.5cm}
\vspace{3mm}
{ \normalsize JCOM on Chicago}
\\

\smallskip
$N=0$, 
$\overline{a}^+=0.0000$\\

\smallskip
No data
\end{minipage}
\end{tabular}\\
\begin{tabular}{c|c|c} \hline 
\begin{minipage}[t]{4.5cm}
\vspace{3mm}
{ \normalsize JCOM on NYSE}
\\

\smallskip
$N=0$, 
$\overline{a}^+=0.0000$\\

\smallskip
No data
\end{minipage}
&\begin{minipage}[t]{4.5cm}
\vspace{3mm}
{ \normalsize JCOM on ARCA}
\\

\smallskip
$N=5000$, 
$\overline{a}^+=3.0464$\\

\smallskip
$P_S=0.024$ $P_m=0.039$
\\

 \smallskip
$\kappa_0=0.9528(0.0247)^{***}$\\$\kappa_1=0.7065(0.0241)^{***}$\\$\alpha_\kappa=-0.1111(0.0271)^{***}$\\$\rho_0=1.5036(0.0341)^{***}$\\$\rho_1=1.1213(0.0359)^{***}$\\$\alpha_\rho=-0.2881(0.0260)^{***}$\\\smallskip
\end{minipage}
&\begin{minipage}[t]{4.5cm}
\vspace{3mm}
{ \normalsize JCOM on NASDAQ T}
\\

\smallskip
$N=0$, 
$\overline{a}^+=0.0000$\\

\smallskip
No data
\end{minipage}
\end{tabular}\\
\begin{tabular}{c|c|c} \hline 
\begin{minipage}[t]{4.5cm}
\vspace{3mm}
{ \normalsize JCOM on CBOE}
\\

\smallskip
$N=118$, 
$\overline{a}^+=42.8000$\\

\smallskip
$P_S=0.270$ $P_m=0.251$
\\

 \smallskip
$\kappa_0=0.0250(0.0021)^{***}$\\$\alpha_\kappa=-0.3670(0.0301)^{***}$\\$\rho_0=0.0362(0.0044)^{***}$\\$\alpha_\rho=-0.1593(0.0367)^{***}$\\\smallskip
\end{minipage}
&\begin{minipage}[t]{4.5cm}
\vspace{3mm}
{ \normalsize JCOM on BATS}
\\

\smallskip
$N=5000$, 
$\overline{a}^+=15.0924$\\

\smallskip
$P_S=0.662$ \\

 \smallskip
$\kappa_0=0.0027(0.0000)^{***}$\\$\rho_0=0.0138(0.0001)^{***}$\\\smallskip
\end{minipage}
&\begin{minipage}[t]{4.5cm}
\vspace{3mm}
{ \normalsize JCOM on ISE}
\\

\smallskip
$N=2224$, 
$\overline{a}^+=46.0266$\\

\smallskip
$P_S=-1.287$ \\

 \smallskip
$\kappa_0=0.0007(0.0000)^{***}$\\$\rho_0=0.0048(0.0000)^{***}$\\\smallskip
\end{minipage}
\end{tabular}\\
\begin{tabular}{c|c|c} \hline 
\begin{minipage}[t]{4.5cm}
\vspace{3mm}
{ \normalsize ARL on NASDAQ OB}
\\

\smallskip
$N=0$, 
$\overline{a}^+=0.0000$\\

\smallskip
No data
\end{minipage}
&\begin{minipage}[t]{4.5cm}
\vspace{3mm}
{ \normalsize ARL on NSE}
\\

\smallskip
$N=0$, 
$\overline{a}^+=0.0000$\\

\smallskip
No data
\end{minipage}
&\begin{minipage}[t]{4.5cm}
\vspace{3mm}
{ \normalsize ARL on Chicago}
\\

\smallskip
$N=0$, 
$\overline{a}^+=0.0000$\\

\smallskip
No data
\end{minipage}
\end{tabular}%
\\

\end{document}